\documentclass[10pt,table,dvipsnames]{amsart}
\usepackage[utf8]{inputenc}
\usepackage{amsmath,amsfonts,amssymb,amscd}
\usepackage{diagbox}
\usepackage{subcaption}
\usepackage{graphicx}
\usepackage[dvipsnames]{xcolor}
\usepackage{mathrsfs}
\usepackage{enumitem}
\usepackage{multicol}
\usepackage{xcolor}
\usepackage{pgf, tikz,tikz-3dplot, pgfplots}  
\usetikzlibrary{positioning}
\usepackage{xifthen}
\usepackage[]{algorithm2e}
\usepackage{algorithmicx}
\usepackage{graphicx}
\usepackage{pstricks}
\usepackage{hyperref}
\usepackage{float}
\usepackage{geometry}
\usepackage{mathtools}

\DeclarePairedDelimiter\floor{\lfloor}{\rfloor}

  \newtheorem{theorem}{Theorem}[section]

  \newtheorem{lemma}[theorem]{Lemma}

  \newtheorem{example}[theorem]{Example}
  \newtheorem{remark}[theorem]{Remark}
  \numberwithin{equation}{section}

\def\F{\mathbb{F}}

\def\codeC{\mathscr{C}}

\colorlet{Mycolor1}{green!10!orange!95!}

\definecolor{olive}{rgb}{0.5, 0.5, 0.0}

\definecolor{darkolivegreen}{rgb}{0.33, 0.42, 0.18}

\definecolor{hookergreen}{rgb}{0.0, 0.62, 0.0}

\definecolor{royalblue(web)}{rgb}{0.25, 0.41, 0.88}

\definecolor{rosemadder}{rgb}{0.89, 0.15, 0.21}

	\definecolor{cadet}{rgb}{0.33, 0.41, 0.47}

\title{On Grid Codes}

\author{E. J. García-Claro}
\address{Departamento de Matemáticas, Universidad Autónoma Metropolitana, Unidad Iztapalapa, Código postal 09340, Ciudad de México - México}
\email{eliasjaviergarcia@gmail.com}
\thanks{ }

\author{Ismael Gutierrez $\ddagger$}

\address{Department of Mathematics and Statistics, Universidad del Norte, Km 5 via a Puerto Colombia, Barranquilla - Colombia}
\email{isgutier@uninorte.edu.co}
\thanks{$\ddagger$The second author thanks the Deutscher Akademischer Austauschdienst for the financial support provided.}


\date{ }


\begin{document}
\maketitle

\begin{abstract}

Generating functions for the size of a $r$-sphere, with respect to the Manhattan distance in an $n$-dimensional grid, are used to provide explicit formulas for the minimum and maximum size of an $r$-ball centered at a point of the grid. This allows us to offer versions of the Hamming and Gilbert-Varshamov bounds for codes in these grids. Relations between the Hamming, Manhattan, and Lee distances defined in an abelian group $G$ are studied. A formula for the minimum Hamming distance of codes that are cyclic subgroups of $G$ is presented. Furthermore, several lower bounds for the minimum Manhattan distance of these codes based on their minimum Hamming and Lee distances are established. Examples illustrating the main results are presented, including several SageMath implementations.

\end{abstract}

\section{Introduction}\label{S1}

A  finite $n$-dimensional square grid is a cartesian product of graphs $\square_{i=1}^{n}P_{m_{i}}$ where $P_{m_{i}}$ is the path graph with $m_{i}$ vertices for $i=1,...,n$. These grids appear naturally in several research areas, among which are 
crystallography, and the theoretical study of some materials or digital images. 
In crystallography, these grids come up as part of the called Bravais lattices that contain unit cells encapsulating minimal replicating structures among crystal formations, similar to the mathematical concept of fundamental domain \cite{Cristal1, Cristal2}. Materials with a $3$-dimensional grid shape arise in the study of the structure and properties of porous media; these last have diverse applications such as enhanced oil recovery and faster catalytic reactions, among others \cite{Material1, Material2}. Digital images are formed by arrays of pixels each containing a color that is a combination of a tone of red, green, and blue, represented by an element of $[0,255]^{3}$ (called the RGB representation of the color). Thus any digital image with $m\times n$ pixels can be interpreted as a subset of the set of vertices $[0, m-1]\times[0, n-1]\times [0,255]^{3}$ of the $5$-dimensional grid $P_{m}\square P_{n} \square P_{256}^{\square 3}$, i.e., any digital image is a grid code (see definition below).\\

If $x\in \mathbb{R}^{n}$, $x_{i}$ will denote the $i$-th projection of $x$. The Manhattan distance in a subset of $\mathbb{R}^{n}$ (also called, grid distance and taxicab distance) is defined as $d(x,y)=\sum_{i=1}^{n}|x_{i}-y_{i}|$, this is a Minkowsky distance $l_{p}$ (when $p=1$).  Let $\Gamma$ be the graph with $\mathbb{Z}^{n}$ as set of vertices, in which two vertices are adjacent if the Manhattan distance between them is $1$ (see \cite[page 333]{Deza1}). In graph theory, the path distance between two vertices in a connected graph is defined as the number of edges in a shortest path connecting them. The Manhattan distance in $\mathbb{Z}^{n}$ coincides with the path distance defined in $\Gamma$ (see \cite[page 248]{Deza2}). Thus, since any $n$-dimensional grid can be embeded in $\Gamma$, it can be considered that the Manhattan and the path distances are the same in any $n$-dimensional grid. This metric can be found in diverse real-world applications. For example, in \cite{Greche} the authors studied the Manhattan and Euclidean distances in the classification of facial expressions of basic emotions. In \cite{Li} exponential transformation and tangent distance are used to improve the accuracy of weighted-KNN (KNearest Neighbors) based fingerprint localization, where Manhattan tangent distance (MTD) and approximate Manhattan tangent distance (AMTD) are proposed. Experiments demonstrate that MTD and AMTD outperform common metrics. \\

Let $m_{i}\in \mathbb{Z}_{>1}$ and $[0,m_{i}-1]$ be an the interval of integers for $i=1,...,n$. Then $\mathfrak{G}:=\prod_{i=1}^{n}[0,m_{i}-1]$ is a metric space with the Manhattan distance. This permits to interpret $\mathfrak{G}$ as an $n$-dimensional grid. In fact, if $n=2$, $\mathfrak{G}$ will be the set of the vertices of a rectangular grid, and isometric to this grid with the path distance. If $n=3$, $\mathfrak{G}$ will be the set of the vertices of a Rubik's rectangle, and isometric to the $1$-Skeleton of this with the path distance. In general, if $P_{m_{i}}$ denotes the path graph with $m_{i}$ vertices for $i=1,...,n$, $\mathfrak{G}$ will be isometric to the Cartesian product of graphs $\square_{i=1}^{n}P_{m_{i}}$ with the path distance. A grid code will be defined as a subset of $\mathfrak{G}$ (since $\mathfrak{G}$ is by definition isometric to $\square_{i=1}^{n}P_{m_{i}}$, a grid code could be considered as a subset of the set of vertices of $\square_{i=1}^{n}P_{m_{i}}$ \footnote{Like codes with the Hamming distance, which are subsets of the set vertices of a graph called the Hamming graph.} as well).  Codes (of arbitrary length) with the Hamming distance that are subsets of a Cartesian product of alphabets of distinct sizes have been considered before (e.g. \cite{Forney92}), but not in the context of the Manhattan distance.\\

The Hamming distance is a fundamental concept in coding theory. If $A$ is a finite alphabet (set), the Hamming distance between two words in $A^{n}$ is defined as the number of entries in which two distinct words ($n$-tuples) differ. If $C\subset A^{n}$, $k=log_{|A|}(|C|)$, and $d$ is the minimum distance among two different words in $C$, it is said that $C$ is an $(n, k, d)$-code over $A$. In the past, changing the Hamming distance for a new one has been useful to find new applications. For instance, the Lee distance was first introduced in \cite{Ulrich, Lee} when dealing with the transmission of signals over noisy channels. Ever since, various types of codes with the Lee distance, among which are negacyclic codes \cite{Berlekamp} and perfect error-correcting Lee codes \cite{Golomb}, have been studied (see \cite{Roth, Araujo, Etzion, Etzion2} for other results on Lee codes). The Manhattan distance also offers another alternative to the Hamming distance in the study of codes. However, in spite of being widely used for real-world applications, it does not seem to have been as widely explored as the Lee distance in the context of coding theory. This might be due to the lack of channels matched to the Manhattan distance \cite{Gabidulin}.\\ 

 Some results of the late years on codes with the Manhattan distance are the following: In \cite[Theorem 2.25]{Bevan} is characterized the property of a permutation of being $k$-prolific in terms of its ``breath'' (the minimum Manhattan distance of a code associated with the permutation \cite{Blackburn}). Later, in \cite{Blackburn} probabilistic properties derived from the analysis of the minimum Manhattan distance and jump of a permutation are studied. In \cite{Zeulin} it is proved that remainder codes with the Manhattan distance have a greater rate compared to similar codes with the Hamming distance, for large power of the input alphabet. In \cite{Sok} it is shown that deletion codes for the editing metric are equivalent to codes over the integers with the Manhattan metric by run length coding. This is later applied to give bounds on the maximum size of the studied codes.\\

  Determining bounds for the maximum size of a code with a prescribed minimum distance is a problem of interest in coding theory.  Some well-known results of that kind are the Hamming and the Gilbert-Varshamov\footnote{The Gilbert-Varshamov bound is sometimes refereed simply as the Gilbert bound.} bounds. For codes with the Hamming distance, these bounds can be found in many textbooks of coding theory (e.g. \cite{Roth,Huffman,Richardson}). The Hamming bound for Lee codes was introduced in \cite{Golomb}. In \cite{Atola82} an Elias type bound was presented and in \cite{ Atola84} asymptotic versions of the Hamming and Gilbert-Varshamov bounds were given for Lee codes.\\
  
The first aim of this work is to introduce the Hamming and Gilbert-Varshamov bounds for grid codes. The second one is to present several lower bounds for the minimum Manhattan distance of codes that are cyclic subgroups of an abelian group. The manuscript is organized as follows. In Section \ref{S2}, some preliminary concepts that will be needed later are presented. In Sections \ref{S3} and \ref{S4}, alternative versions of the Hamming and Gilbert-Varshamov bounds for grid codes are given. These bounds depend on the computation of the largest and smallest size of an $r$-ball in an $n$-dimensional grid. In Section \ref{S5}, the concept of local distance enumerator polynomial is introduced as the generating function for the size of an $r$-sphere in an $n$-dimensional grid. This polynomial allows us to provide explicit formulas for the largest and smallest size of an $r$-ball. Finally, in Section \ref{S6}, the minimum Hamming distance is determined for grid codes that are cyclic subgroups of an abelian group, and some bounds for the minimum Manhattan distance of these codes are offered in terms of their minimum Hamming and Lee distances.

\section{Preliminaries} \label{S2}

From now on $r\in \mathbb{Z}_{\geq 0}$ and $[n]:=\{1,...,n\}$. The Lee distance can be defined in $\mathfrak{G}$ as $d_{L}(x,y)=\sum_{i=1}^{n}\min\{|x_{i}-y_{i}|, m_{i}-|x_{i}-y_{i}|\}$. The Hamming ($d_{H}$), Lee ($d_{L}$), and Manhattan ($d$) distances in $\mathfrak{G}$ satisfy $d_{H}(x,y)\leq d_{L}(x,y)\leq d(x,y)$ $\forall x,y \in \mathfrak{G}$; and if $m_{i}=2$ for $i\in [n]$, all three distances coincide. For $x\in \mathfrak{G}$, $B_{r}(x), B_{r}^{L}(x)$ and $ B_{r}^{H}(x)$ will denote the closed $r$-balls centered in $x$ with respect to the Manhattan, Lee and Hamming distances, respectively. In general, $B_{r}(x)\subseteq B_{r}^{L}(x)\subseteq B_{r}^{H}(x)$ $\forall x\in \mathfrak{G}$. A key difference between $d_{H}$ ($d_{L}$) and $d$ is that, the size of an $r$-ball with respect to $d$ depends on the center (see \cite{Gabidulin, Deza1, Deza2} for further properties of these distances).\\ 

If $\codeC\subseteq \mathfrak{G}$ with $|\codeC|\geq 2$, and $d(\codeC)=\min\{d(g,h): g,h\in \codeC \wedge \, g\neq h \}$, it will be said that $\codeC$ is an $(n,|\codeC|,d(\codeC))$-grid code (or simply a code), and that  $n$, $|\codeC|$, $d(\codeC)$ are its parameters; $n$ and $d(\codeC)$ will be called \textit{lenght} and \textit{minimum distance} of $\codeC$. In what remains $\codeC$ will denote a code in $\mathfrak{G}$, unless stated otherwise.\\


\section{Hamming bound}\label{S3}
The Hamming bound in the context of grid codes is studied in this section. Let $\codeC$ be a code in $\mathfrak{G}$. The \textit{packing radius} of $\codeC$ is the largest (non-negative integer) $t$ such that the $t$-balls centered at codewords of $\codeC$ are disjoint (i.e., $\cup_{c\in \codeC } B_{t}(c)=\sqcup_{c\in \codeC} B_{t}(c)$). It is easy to check that the \textit{packing radius} of a code with minimum distance $d$ is $t=\lfloor \frac{d-1}{2} \rfloor$ (independently of the metric). The \textit{covering radius} of $\codeC$ is the smallest (non-negative integer) $s$ such that the set of the $s$-balls centered at codewords of $\codeC$ cover $\mathfrak{G}$ (i.e., $\mathfrak{G}=\cup_{c\in \codeC } B_{s}(c)$).  $\codeC$ is a \textit{perfect code} if there exists a radius $t$ such that the $t$-balls centered at elements of $\codeC$ form a partition of $\mathfrak{G}$, i.e., its packing radius and covering radius coincide.\\

Let $\mathcal{A}_{\mathfrak{G}}(n,d):=\max \{|\codeC|: \codeC \text{ is a } (n,M ,d')\text{-code in } \mathfrak{G} \wedge d'\geq d\}$. Let  $\codeC\subseteq \mathfrak{G}$, then $\eta_{r}(\codeC):=\min \{|B_{r}(c)| : c\in \codeC\}$.

\begin{theorem}[Hamming bound]\label{spb}
Let  $1 \leq d$ and $t:=\lfloor \frac{d-1}{2} \rfloor$. Then 

\[\mathcal{A}_{\mathfrak{G}}(n,d)\leq \dfrac{\prod_{i=1}^{n}m_{i}}{\eta_{t}(\mathfrak{G})}.\]

In particular, if $\codeC$ is an $(n, M, d')$-code in $\mathfrak{G}$ with $1 \leq d\leq d'$ and $|\codeC|=\left \lfloor \dfrac{\prod_{i=1}^{n}m_{i}}{\eta_{t}(\mathfrak{G})}   \right \rfloor$, then $|\codeC|=\mathcal{A}_{\mathfrak{G}}(n,d)$.
\end{theorem}
 
\begin{proof}
Let $\codeC$ be an $(n, M, d')$-code in $\mathfrak{G}$ with $1 \leq d\leq d'$.
By definition of $t$, the $t$-balls centered at elements of $\codeC$ are disjoint. Thus, since $\eta_{t}(\mathfrak{G})\leq \eta_{t}(\codeC)$ and $\sqcup_{c\in \codeC}B_{t}(c)\subseteq \mathfrak{G}$, then

\[
|\codeC |\eta_{t}(\mathfrak{G}) \leq |\codeC |\eta_{t}(\codeC)= \sum_{c\in \codeC}\eta_{t}(\codeC) \leq   \sum_{c\in \codeC} |B_{t}(c)|=\Bigl|\sqcup_{c\in \codeC}B_{t}(c)\Bigr|\leq  |\mathfrak{G}|=\prod_{i=1}^{n}m_{i},
\]

implying that $\mathcal{A}_{\mathfrak{G}}(n,d)\leq \frac{\prod_{i=1}^{n}m_{i}}{\eta_{t}(\mathfrak{G})}$. The rest follows from the fact that $|\codeC|\leq\mathcal{A}_{\mathfrak{G}}(n,d) \leq \left \lfloor \frac{\prod_{i=1}^{n}m_{i}}{\eta_{t}(\mathfrak{G})}\right \rfloor$.  
\end{proof}

 Theorem \ref{spb} implies that, for $t\in \mathbb{Z}_{>0}$, any set of vertices in an $n$-dimensional grid (i.e., a grid code) with size greater than $\frac{\prod_{i=1}^{n}m_{i}}{\eta_{t}(\mathfrak{G})}$ has minimum distance  less than $d=2t+1 \text{ or } 2t+2$.

\begin{example}\label{experfect}
Let $m_{1}=5$, $m_{2}=2$. Then $\mathfrak{G}=\prod_{i=1}^{2}[0,m_{i}-1]=[0,4]\times [0,1]$.\\

If $\codeC=\{00,41\}$, then $\codeC$ is a $(2,2,5)$-code with packing radius $t=\lfloor \frac{5-1}{2} \rfloor=2$, and $\{|B_{2}(g)|: g\in \mathfrak{G}\}=\{5,6,7\}$. Thus, by Theorem \ref{spb}, $2=|\codeC|\leq \mathcal{A}_{\mathfrak{G}}(2,5)\leq \frac{\prod_{i=1}^{2}m_{i}}{\eta_{2}(\mathfrak{G})}=\frac{10}{5}=2$, implying that $|\codeC|= \mathcal{A}_{\mathfrak{G}}(2,5)$. In Figure \ref{fig3}, it can be seen that $\codeC$ is a perfect code.\\

\begin{figure}[h!]
\begin{tikzpicture}[scale=2.3, vertices/.style={draw,circle, inner sep=3pt}]

\draw[step=1cm,black, thin] (0,0) grid (4,1);

\draw [thick, gray!80, fill=gray!30, opacity=.5] (0,0)--(2,0)--(1,1)--(0,1)--(0,0);

\draw [ thick, gray!80, fill=gray!30, opacity=.5] (4,1)--(2,1)--(3,0)--(4,0)--cycle;

\node[vertices,  label=below:{$0$}] (a) at (0,0) {};
\node[vertices, label=below:{$1$}] (b) at (1,0) {};
\node[vertices, label=below:{$2$}] (c) at (2,0) {};
\node[vertices, label=below:{$3$}] (d) at (3,0) {};
\node[vertices, label=below:{$4$}] (e) at (4,0) {};

\node[vertices, label=left:{$1$}] (b) at (0,1) {};

\node[vertices, fill=Cerulean!60, circle] at (0,0) {}; 
\node[vertices,fill=Cerulean!60, circle] at (1,0) {}; 
\node[vertices,fill=Cerulean!60, circle] at (2,0) {}; 
\node[vertices,fill=Cerulean!60, circle] at (1,1) {};
\node[vertices,fill=Cerulean!60, circle] at (0,1) {}; 

\node[vertices,fill=Cerulean!60, circle] at (3,0) {}; 
\node[vertices,fill=Cerulean!60, circle] at (4,0) {}; 
\node[vertices,fill=Cerulean!60, circle] at (2,1) {}; 
\node[vertices,fill=Cerulean!60, circle] at (3,1) {};
\node[vertices,fill=Cerulean!60, circle, label=above:{$41$}] at (4,1) {};

\end{tikzpicture}

\caption{The $2$-balls centered at $\codeC=\{00,41\}$ cover $\mathfrak{G}$.} 
\label{fig3}
\end{figure}

 If $\codeC' =\{01,20,41\}$, then $\codeC'$ is a $(2,3,3)$-code in $\mathfrak{G}$ with packing radius $t=\lfloor \frac{3-1}{2}\rfloor=1$, and $\{|B_{1}(g)|: g\in \mathfrak{G} \}=\{3,4\}$. Thus, by Theorem \ref{spb}, $3=|\codeC' |\leq \mathcal{A}_{\mathfrak{G}}(2,3)\leq \frac{\prod_{i=1}^{2}m_{i}}{\eta_{1}(\mathfrak{G})}=\frac{10}{3}=3.33$, and hence $|\codeC'|= \mathcal{A}_{\mathfrak{G}}(2,3)$. In Figure \ref{fig4}, it can be seen that $\codeC'$ is a perfect code. 

\begin{figure}[h!]
\begin{tikzpicture}[scale=2.3, vertices/.style={ draw, circle, inner sep=3pt}]

\node[vertices,  label=below:{$0$}] (a) at (0,0) {};
\node[vertices, label=below:{$1$}] (b) at (1,0) {};
\node[vertices, label=below:{$2$}] (c) at (2,0) {};
\node[vertices, label=below:{$3$}] (d) at (3,0) {};
\node[vertices, label=below:{$4$}] (e) at (4,0) {};

\node[vertices, label=left:{$1$}] (b) at (0,1) {};

\draw[step=1cm,black, thin] (0,0) grid (4,1);

\draw [  thick, gray!80, fill=gray!30, opacity=.5] (0,0)--(0,1)--(1,1)--(0,0);

\draw [  thick, gray!80, fill=gray!30, opacity=.5] (1,0)--(2,1)--(3,0)--(1,0);

\draw [  thick, gray!80, fill=gray!30, opacity=.5] (4,0)--(3,1)--(4,1)--(4,0);

\node[vertices, fill=Cerulean!60, circle] at (0,0) {};
\node[vertices,fill=Cerulean!60, circle] at (0,1) {};
\node[vertices,fill=Cerulean!60, circle] at (1,1) {}; 

\node[vertices,fill=Cerulean!60, circle] at (1,0) {};
\node[vertices,fill=Cerulean!60, circle] at (2,0) {};
\node[vertices,fill=Cerulean!60, circle] at (3,0) {}; 
\node[vertices,fill=Cerulean!60, circle] at (2,1) {};

\node[vertices,fill=Cerulean!60, circle] at (4,0) {};
\node[vertices,fill=Cerulean!60, circle] at (3,1) {};
\node[vertices,fill=Cerulean!60, circle, label=above:{$41$}] at (4,1) {};

\end{tikzpicture}

\caption{The $1$-balls centered at $\codeC'=\{01,20,41\}$ cover $\mathfrak{G}$. }
\label{fig4}
\end{figure}

\end{example}

A code $\codeC$ with packing radius $t$ will be said to attain the Hamming bound if $|\codeC|=\frac{\prod_{i=1}^{n}m_{i}}{\eta_{t}(\mathfrak{G})}$, i.e., $\prod_{i=1}^{n}m_{i}=|\codeC|\eta_{t}(\mathfrak{G})$.
Codes attaining the Hamming bound are perfect, because the $t$-balls centered at their elements are always disjoint, and the equality $\prod_{i=1}^{n}m_{i}=|\codeC|\eta_{t}(\mathfrak{G})$ tells us that they cover all $\mathfrak{G}$; moreover, for these codes, all the $t$-balls have the same size. Since $|\codeC|=2$ in Example \ref{experfect}, $\codeC$ attains the Hamming bound and is a perfect code (as can be seen in Figure \ref{fig3}). In classic coding theory, the converse to that statement is also true, i.e., a code is perfect if and only if it attains the classic Hamming bound, but in this case, that is not true. For instance,  $\codeC'$ in Example \ref{experfect} is perfect (as can be seen in Figure \ref{fig4}), but it does not attain the Hamming bound. In the classic context is also obvious that a perfect code has the biggest possible size for a code with its length and minimum distance. This fact is not obvious here though, because perfect codes determine tessellations of $\mathfrak{G}$, but since the balls (tiles) can have different sizes (depending on the centers), one could imagine that it might be possible to find two tessellations with different sizes, and therefore two perfect codes (formed by the centers of the tiles) with different sizes. To find such an example (or prove its non-existence) is still an open problem.\\

A trivial perfect code in $\mathfrak{G}$ will be a perfect code with packing radius $t=0$, $t=\lfloor \frac{\sum_{i=1}^{n}(m_{i}-1)}{2} \rfloor$ or $t=\sum_{i=1}^{n}(m_{i}-1)$. In some cases, $\mathfrak{G}$ does not have perfect proper codes.  For example, in $\mathfrak{G}=[0,2]^{2}$, the only possible packing radius are $t=0,1$ (because any pair of $2$-ball have non-trivial intersection). In addition, it is easy to see that there are no perfect codes with packing radius $t=1$ so that $\mathfrak{G}$ is the only perfect code in $\mathfrak{G}$. Hence it is natural to ask under what conditions on $\mathfrak{G}$ and a given $t\in \mathbb{Z}_{>0}$, there would exist a perfect code in $\mathfrak{G}$ having $t$ as its packing radius. In some cases the code $\codeC=\{(0,...,0), (m_{1}-1,m_{2}-1,...,m_{n}-1)\}$ is a trivial perfect code of packing radius $t=\lfloor \frac{\sum_{i=1}^{n}(m_{i}-1)}{2} \rfloor$, such is the case of $\codeC$ in Example \ref{experfect}.

\section{Gilbert-Varshamov bound} \label{S4}

An alternative version of the Gilbert–Varshamov bound is presented in this section.\\

Let  $\codeC\subseteq \mathfrak{G}$, then $\gamma_{r}(\codeC):=\max \{|B_{r}(c)| : c\in \codeC\}$.

\begin{remark} \label{rmark1}
In \cite[Theorem 4]{Golomb} it is stated that $|B_{r}^{L}(x)|=\sum_{j=0}^{\min\{r,n\}}2^{j}\binom{n}{j}\binom{r}{j}$  $\forall x\in \mathbb{Z}^{n}$, which is false. For example, in $\mathbb{Z}_{4}^{2}$, $|B_{3}^{L}(00)|=15<\sum_{j=0}^{\min\{3,2\}}2^{j}\binom{2}{j}\binom{3}{j}=25$. However, by checking the proof of \cite[Theorem 4]{Golomb} one may verify that the formula given for $|B_{r}^{L}(x)|$ is in reality a formula for the size of the $r$-ball $\textbf{B}_{r}(x)$ in $\mathbb{Z}^{n}$. Thus since $B_{r}(x)\subseteq \textbf{B}_{r}(x)$ $\forall x\in \mathfrak{G}$, if $x\in \mathfrak{G}$, $\gamma_{r}(\mathfrak{G})
\leq |\textbf{B}_{r}(x)|=\sum_{j=0}^{\min\{r,n\}}2^{j}\binom{n}{j}\binom{r}{j}$.
\end{remark}


If $\codeC$ is a code such that the $r-$balls centered at words of $\codeC$ cover $\mathfrak{G}$, it will be said that $\codeC$ is an $r$-covering code (over $\mathfrak{G}$). It is easy to check that if  $\codeC$ is an $(n, M, d)$-code, $|\codeC|=\mathcal{A}_{\mathfrak{G}}(n,d)$ implies that $\codeC$ is an $(d-1)$-covering code, but the converse is not always true.

\begin{theorem}[Gilbert–Varshamov bound]\label{gvb}

\[ \dfrac{\prod_{i=1}^{n}m_{i}}{\sum_{j=0}^{\min\{d-1,n\}}2^{j}\binom{n}{j}\binom{d-1}{j}} \leq\dfrac{\prod_{i=1}^{n}m_{i}}{\gamma_{d-1}(\mathfrak{G})} \leq \mathcal{A}_{\mathfrak{G}}(n,d).\]
\end{theorem}

\begin{proof}
 Let $\codeC$ be an $(n,\mathcal{A}_{\mathfrak{G}}(n,d),d)$-code  in $\mathfrak{G}$. Then $\mathfrak{G}=\cup_{c\in \codeC}B_{d-1}(c)$. Otherwise there exist $c_{0}\in\mathfrak{G}-\cup_{c\in \codeC}B_{d-1}(c)$, and so the code $\codeC\cup\{c_{0}\}$ has minimum distance greater than or equal $d$, which contradicts that $|\codeC|=\mathcal{A}_{\mathfrak{G}}(n,d)$. Thus 

\[
 \prod_{i=1}^{n}m_{i}=|\mathfrak{G}|= \Bigl| \cup_{c\in \codeC}B_{d-1}(c) \Bigr| \leq \sum_{c\in \codeC} |B_{d-1}(c)| \leq \sum_{c\in \codeC} \gamma_{d-1}(\codeC)=|\codeC |\gamma_{d-1}(\codeC),
\]
implying that  $\dfrac{\prod_{i=1}^{n}m_{i}}{\sum_{j=0}^{\min\{d-1,n\}}2^{j}\binom{n}{j}\binom{d-1}{j}} \leq \dfrac{\prod_{i=1}^{n}m_{i}}{\gamma_{d-1}(\mathfrak{G})} \leq \dfrac{\prod_{i=1}^{n}m_{i}}{\gamma_{d-1}(\codeC)} \leq |\codeC| = \mathcal{A}_{\mathfrak{G}}(n,d)$ where the first inequality is by Remark \ref{rmark1}, and the second one is because $\gamma_{d-1}(\codeC)\leq \gamma_{d-1}(\mathfrak{G})$.
\end{proof}

Since the bounds in Theorem \ref{gvb} are also bounds for the size of any $(n, M, d)$-code that is an $(d-1)$-covering code, any $(n, M, d)$-code with less than $\frac{\prod_{i=1}^{n}m_{i}}{\gamma_{d-1}(\mathfrak{G})}$ can not be an $(d-1)$-covering code nor has maximum size. This argument could be helpful to determine if a set of vertices with minimum distance $d$ in an $n$-dimensional grid (i.e., a grid code) is not an $r$-covering code for $r<d$ (i.e., has covering radius greater than or equal $d$). A similar ``greedy'' algorithm as the one used in classic coding theory can be used for producing a grid code with minimum distance at least $d$ that meets the Gilbert-Varshamov bound \cite[see p. 87]{Huffman}. 

\section{Size of an $r$-ball in $\mathfrak{G}$} \label{S5}

In this section, we introduce the concept of \textbf{local distance enumerator polynomial} in the language of generating functions and present a way to compute the size of $r$-ball in $\mathfrak{G}$ in terms of this kind of polynomial. For an introduction to generating functions see, e.g.,  \cite{tucker}.\\

Let $\codeC_{0}$ and $\codeC_{1}$ be subsets of $\mathfrak{G}$. Consider the polynomial $\sum_{(x,y)\in \codeC_{0} \times \codeC_{1}} t^{d(x,y)}$. If $\codeC_{0}$ is a code (i.e., $|\codeC_{0}|\geq 2$), and $\codeC_{0}=\codeC_{1}$ or $\codeC_{1}=\{\textbf{0}\}$, this polynomial becomes an alternative version of the classic distance enumerator polynomial or weight enumerator polynomial of $\codeC_{0}$, respectively. In addition, if $\codeC_{0}= \mathfrak{G}$ and $\codeC_{1}=\{a\}$, then one gets the polynomial $p(t,a):=\sum_{(x,y)\in \mathfrak{G}  \times \{a\}} t^{d(x,y)}=\sum_{x \in \mathfrak{G}} t^{d(x,a)}=\sum_{j=0}^{\partial}|S_{j}(a)| t^{j}$ with $\partial=deg(p(t,a))$. This polynomial $p(t,a)$ will be called the \textbf{\textit{local distance enumerator polynomial}} of $a$ in $\mathfrak{G}$. Theorem \ref{local} offers a description of $p(t,a)$ that will be useful later for describing its coefficients.

\begin{theorem}\label{local}
Let $a\in \mathfrak{G}=\prod_{i=1}^{n}[0, m_{i}-1]$, $l_{i}(a):= max\{ a_{i}, m_{i}- (a_{i}+ 1)\}$ for all $i$. Then 
\[ p(t,a)= \prod_{i=1}^{n} \left( \sum_{j=0}^{l_{i}(a)}|S_{j}(a_{i})|t^{j} \right).\]
 Moreover, if $0\leq r\leq \sum_{i=1}^{n}l_{i}(a)$ and  \[p(t,a)_{\leq r}:=\text{ the sum of the monomials of degree } \leq r \text{ of } p(t,a),\]
then $p(1,a)_{\leq r}=|B_{r}(a)|$.
\end{theorem}

\begin{proof}
Note that $\partial$ is equal to
\begin{eqnarray*}
 \max\{d(x,a):x\in \mathfrak{G}\}&=&\max \left( \left\{\sum_{i=1}^{n}|x_{i}-a_{i}|:x_{i}\in [0,m_{i}-1]\, \forall \, i\in[1, n]\right\}\right)\\
              &=&  \max \left( \sum_{i=1}^{n} \left\{ |x_{i}-a_{i}|: x_{i}\in [0,m_{i}-1] \right\} \right)\\ 
              &=&\sum_{i=1}^{n} \max \left\{ |x_{i}-a_{i}|:x_{i}\in [0,m_{i}-1] \right\} =\sum_{i=1}^{n}l_{i}(a).
\end{eqnarray*}

Now it is enough to show that the coefficient $b_{r}$ of degree $r$ of the polynomial $b(t):=\prod_{i=1}^{n} \left( \sum_{j=0}^{l_{i}(a)}|S_{j}(a_{i})|t^{j} \right)$ is equal to $|S_{r}(a)|$ for $r=0,...,\sum_{i=1}^{n}l_{i}(a)$. We will proceed by induction over $n$. If $n = 1$, the statement is clear. Let $n = 2$, and 

\[h(x,y,a)=\left( \sum_{j=0}^{l_{1}(a)}|S_{j}(a_{1})|x^{j} \right)\cdot \left( \sum_{u=0}^{l_{2}(a)}|S_{u}(a_{2})|y^{u} \right)= \sum_{j=0}^{l_{1}(a)} \sum_{u=0}^{l_{2}(a)}|S_{j}(a_{1})||S_{u}(a_{2})|x^{j}y^{u}.
\]


 Since $S_{j}(a_{1})\times S_{u}(a_{2})=\{(x,y)\in \mathfrak{G}:  |x-a_{1}|=j \wedge |y-a_{2}|=u \}$ for all $j$ and $u$, then $S_{r}(a)= \bigsqcup_{j+u=r} S_{j}(a_{1})\times S_{u}(a_{2})$ for all $r\in [0,l_{1}+l_{2}]$. Thus $|S_{r}(a)|= \sum_{j+u=d} |S_{j}(a_{1})||S_{u}(a_{2})|$ for all $r\in [0,l_{1}+l_{2}]$ . Therefore, $h(t,t,a)=b(t)$ is such that $b_{r}=|S_{r}(a)|$. Suppose the statement is true for $n=k$. Let $n=k+1$ and $a=(a_{1},...,a_{k+1})\in \mathfrak{G}$ , then 
\[ p(t, a):= \prod_{i=1}^{k+1} \left( \sum_{j=0}^{l_{i}(a)}|S_{j}(a_{i})|t^{j} \right)= q(t, \widehat{a})\cdot \left( \sum_{j=0}^{l_{k+1}(a)}|S_{j}(a_{k+1})|t^{j} \right) \]
where  $\widehat{a}=(a_{1},...,a_{k})$ and $q(t, \widehat{a})= \prod_{i=1}^{k} \left( \sum_{j=0}^{l_{i}(a)}|S_{j}(a_{i})|t^{j} \right)$. By induction hypothesis, the coefficient $q_{s}$ of degree $s$ of $q(t,\widehat{a})$ is equal to $|S_{s}(\widehat{a})|$  for all $s$. Thus $b_{r}=\sum_{s+u=r}q_{s}|S_{u}(a_{k+1})|=\sum_{s+u=r}|S_{s}(\widehat{a})||S_{u}(a_{k+1})|=|S_{r}(a)|$, because $S_{r}(a)=\bigsqcup_{s+u=d}S_{s}(\widehat{a})\times S_{u}(a_{k+1})$. The rest follows from the fact that $B_{r}(a)=\sqcup_{k=0}^{r} S_{k}(a)$.

\end{proof}

\begin{example}

For example, if $\mathfrak{G}=[0,3]^{2}$, then $a=(a_{1},a_{2})=(1,1)\in \mathfrak{G}$ is such that $|S_{0}(a_{1})|=1$, $|S_{1}(a_{1})|=2$ and $|S_{2}(a_{1})|=1$, as can be deduced from the following graph:

\begin{figure}[h!]
\centering
\begin{tikzpicture}
[scale=2, vertices/.style={draw, thick, fill=gray!50, circle, inner sep=3pt}]

\node[vertices, label=below:{$0$}] (a) at (0,0) {};
\node[vertices,label=below:{$a_{1}=1$}, fill=black] (b) at (1,0) {};
\node[vertices,label=below:{$2$}] (c) at (2,0) {};

\node[vertices,label=below:{$3$}] (d) at (3,0) {};

\foreach \to/\from in
{a/b,b/c,c/d} \draw [-] (\to)--(\from);

\end{tikzpicture}
\end{figure}

So that  $|S_{0}(a_{1})|+|S_{1}(a_{1})|t+|S_{2}(a_{1})|t^{2}=1+2t+t^{2}$ is the first of the factors appearing in the description of $p(t, a)$ given in Theorem \ref{local}. On the other hand, since $a_{1}=a_{2}=1$, the first and the second factor of $p(t,a)$ are equal, and  $p(t,a)=(1+2t+t^{2})^{2}$.  In Figure \ref{fig5},  the one-colored dots represent an  $r$-sphere centered at $a=(1,1)$ (in black) and it can be seen that there are as many of these as the coefficient $p_{r}$ (of degree $r$) of  $p(t, a)$ (represented with the same color).

\begin{figure}[h!]
\begin{tikzpicture}[scale=1.3, vertices/.style={draw, fill=WildStrawberry!70, circle, inner sep=3pt}]

\begin{scope}[xshift=3cm]
 
 \draw[fill=gray!10] (0,0) -- (0,3) -- (3,3) -- (3,0) -- cycle; 
  
 \draw[step=1cm, black,very thin] (0,0) grid (3,3);

\node[vertices,fill=black]  at (1,1) {};
\node[vertices,fill=ForestGreen]  at (0,1 ) {};
\node[vertices,fill=ForestGreen]  at (1,0 ) {};
\node[vertices,fill=ForestGreen]  at (1,2 ) {};
\node[vertices,fill=ForestGreen]  at (2,1 ) {};

\node[vertices,fill=red]  at (0, 0) {};
\node[vertices,fill=red]  at (2, 0) {};
\node[vertices,fill=red]  at (3, 1) {};
\node[vertices,fill=red]  at (2, 2) {};
\node[vertices,fill=red]  at (1, 3) {};
\node[vertices,fill=red]  at (0, 2) {};

\node[vertices,fill=YellowOrange]  at (3,0) {};
\node[vertices,fill=YellowOrange]  at (3,2) {};
\node[vertices,fill=YellowOrange]  at (2,3) {};
\node[vertices,fill=YellowOrange]  at (0,3) {};

\node[vertices,fill=Cerulean!60]  at (3,3) {};

\end{scope}

\end{tikzpicture}
\caption{If $\mathfrak{G}=[0,3]^2$, the local distance enumerator polynomial of  $a=(1,1)$ in $\mathfrak{G}$  is $p(t,a)=(1+2t+t^2)^{2}=\textbf{\textcolor{black}{1}}+\textbf{\textcolor{ForestGreen}{4}}t+\textbf{\textcolor{red}{6}}t^2+\textbf{\textcolor{YellowOrange}{4}}t^3+\textbf{\textcolor{Cerulean}{1}}t^4$. Hence $p(1,a)_{\leq 2}=\textbf{\textcolor{black}{1}}+\textbf{\textcolor{ForestGreen}{4}}\cdot 1+\textbf{\textcolor{red}{6}}\cdot 1^2=11=|B_{2}(a)|$.} 
\label{fig5}
\end{figure}
\end{example}

To describe $\eta_{r}(\mathfrak{G})$ ($\gamma_{r}(\mathfrak{G})$) it is sufficient to center an $r$-ball at an element $x$ of $ \mathfrak{G}$ that minimize (maximize) $|B_{r}(x)|$, and provide a formula to calculate that size. We introduce the concept of outermost and innermost elements of $\mathfrak{G}$ for that purpose.   The set of outermost (or corner) elements of $\mathfrak{G}$ will be defined as
\[ Otm(\mathfrak{G})=\prod_{i=1}^{n} \left\lbrace 0 , m_{i}-1  \right\rbrace \subseteq \mathfrak{G}.\]
By construction, $|Otm(\mathfrak{G})|=2^{n}$. The set of innermost elements of $\mathfrak{G}$ will be defined as
 \[Inm(\mathfrak{G})=\prod_{i=1}^{n}\left\lbrace \left\lfloor \dfrac{m_{i}-1}{2} \right\rfloor , \left\lceil \frac{m_{i}-1}{2} \right\rceil \right\rbrace\subseteq \mathfrak{G}.\]
By construction, if $E=\{i\in [n]: 2\mid m_{i}\}$, then $|Inm(\mathfrak{G})|=2^{|E|}$. Note that $\mathfrak{G}=Inm(\mathfrak{G})=Otm(\mathfrak{G})$ if $m_{i}=2$ for all $i\in[n]$\footnote{In this case $\mathfrak{G}$ is the $n$-dimensional hypercube and the Manhattan, Lee and Hamming distances in $\mathfrak{G}$ coincide.}, and that $\mathfrak{G}=Inm(\mathfrak{G})\sqcup Otm(\mathfrak{G})$ if $m_{i}=3$ for all $i\in[n]$.

\begin{lemma} \label{1balls}
Let $l(x):=max\{ x, m-(x+1)\}$ and $l'(x):=min\{ x, m-(x+1)\}$ for all $x\in [0,m-1]$. Let $m\geq2$, $w\in Otm([0,m-1])$, and $z\in  Inm([0,m-1])$; if $m>3$, let $y\notin \left(Otm([0,m-1])\sqcup Inm([0,m-1])\right)$. Then 
\[ |B_{r}(w)|= 
   \begin{cases}  
   1+r &\mbox{if } 0\leq r \leq m-1\\
   m &\mbox{if } m-1<r 
  \end{cases},  
  |B_{r}(z)|= 
   \begin{cases}  
   1+2r &\mbox{if } 0\leq r \leq \lfloor (m-1)/2 \rfloor\\
   m &\mbox{if } \lfloor (m-1)/2 \rfloor<r 
  \end{cases} 
  \]
  and  
  \[|B_{r}(y)|= 
   \begin{cases}  
   1+2r &\mbox{if } 0\leq r \leq l'(y)\\
   1+2l'(y)+ (r-l'(y)) &\mbox{if } l'(y)<r \leq l(y)\\
   m &\mbox{if } l(y)<r 
  \end{cases}.\]
\end{lemma}

\begin{proof}
 Note that 
 \[|S_{j}(w)|= 
   \begin{cases}  
   1 &\mbox{if } 0\leq j \leq m-1\\
   0 &\mbox{if } m-1<j 
  \end{cases},
   |S_{j}(z)|= 
   \begin{cases}  
   1 &\mbox{if } j=0   \\
   2 &\mbox{if } 0< j < \lfloor (m-1)/2 \rfloor\\
   1 &\mbox{if } j=\lceil (m-1)/2 \rceil \wedge 2\mid m \\
   2 &\mbox{if } j=\lceil (m-1)/2\rceil \wedge 2\nmid m \\
   0 &\mbox{if } \lceil (m-1)/2 \rceil <j\\
  \end{cases} \]
  and 
  \[|S_{j}(y)|= 
   \begin{cases}  
   1 &\mbox{if } j=0   \\
   2 &\mbox{if } 0<j \leq l'(y)\\
   1 &\mbox{if } l'(y)<j \leq l(y)\\
   0 &\mbox{if } l(y)<j 
  \end{cases}.\]
The rest follows from the fact that $B_{r}(x)=\sqcup_{j=0}^{r} S_{j}(x)$ for all $x\in[0, m-1]$.
\end{proof}

\begin{theorem}\label{fundamental} Let $a\in \mathfrak{G}$. Then the following statements hold:\\

\begin{enumerate} 

\item If $a \in Otm(\mathfrak{G})$, then  $|B_{r}(a)|=\eta_{r}(\mathfrak{G})$.\\

\item If $a \in Inm(\mathfrak{G})$, then $|B_{r}(a)|=\gamma_{r}(\mathfrak{G})$.

\end{enumerate}

\end{theorem}

\begin{proof}


Let $x\in \mathfrak{G}$. If $n\geq 2$, then $B_{r}(x)=\sqcup_{j=0}^{r}\{a\in B_{r}(x): |a_{n}-x_{n}|=r-j	\}$ where the $j$-th set in this union is equipotent to $B_{j}((x_{i})_{i=1}^{n-1})\subseteq \prod_{i=1}^{n-1}[0,m_{i}-1]$. Therefore $|B_{r}(x)|=\sum_{j=0}^{r}|B_{j}((x_{i})_{i=1}^{n-1})|$.\\

\begin{enumerate}

\item Suppose that  $a\in Otm(\mathfrak{G})$. We will proceed by induction over $n$.  If $n=1$ and $m_{1}\geq 2$,  $|B_{r}(a)|\leq |B_{r}(x)|$ for all $x\in \mathfrak{G}=[0,m_{1}-1]$ (by Lemma \ref{1balls}) and so $|B_{r}(a)|= \eta_{r}(\mathfrak{G})$. Suppose that the statement is true for $k<n$. Since $(a_{i})_{i=1}^{n-1}\in Otm(\prod_{i=1}^{n-1}[0,m_{i}-1])$, then $|B_{r}(a)|=\sum_{j=0}^{r}|B_{j}((a_{i})_{i=1}^{n-1})|$ is minimum (by induction hypothesis), i.e., $|B_{r}(a)|=\eta_{r}(\mathfrak{G})$.\\
 
\item Suppose that  $a\in Inm(\mathfrak{G})$. We will proceed by induction over $n$. If $n=1$ and $m_{1}\geq 2$,  $|B_{r}(x)| \leq |B_{r}(a)| $ for all $x\in \mathfrak{G}=[0,m_{1}-1]$ (by Lemma \ref{1balls})  and so $|B_{r}(a)|= \gamma_{r}(\mathfrak{G})$. Suppose that the statement is true for $k<n$. Since $(a_{i})_{i=1}^{n-1}\in Inm(\prod_{i=1}^{n-1}[0,m_{i}-1])$, then $|B_{r}(a)|=\sum_{j=0}^{r}|B_{j}((a_{i})_{i=1}^{n-1})|$ is maximum (by induction hypothesis), i.e., $|B_{r}(a)|=\gamma_{r}(\mathfrak{G})$. 

\end{enumerate}

\end{proof}

The following remark will be of use for Theorems \ref{etafor} and \ref{gammafor}. Its proof is an elementary application of mathematical induction. 

\begin{remark}\label{binomial}
Let $a_{i}$ and $z_{i}$ be elements in a ring $R$ for $i\in[n]$. Then 

\begin{enumerate}

 \item $\prod_{i=1}^{n}(a_{i}+z_{i})= \sum_{J \in \mathcal{P} ([n])} \left( \prod_{j \in J}a_{j}\right) \cdot \left( \prod_{j \notin  J} z_{j} \right)$. 
 
In particular,

\item $\prod_{i=1}^{n}(1-t^{m_{i}})= \sum_{J \in \mathcal{P} ([n])}(-1)^{|J|} t^{\sum_{i\in J}m_{i}}.$

\item If $y_{i}\in R$, then $\prod_{i=1}^{n}(2y_{i}-1)= \sum_{J \in \mathcal{P} ([n])}(-1)^{n-|J|} 2^{|J|} \prod_{i\in J} y_{i}.$

\end{enumerate}
\end{remark}

The following result offers a general formula for $\eta_{r}(\mathfrak{G})$. 


\begin{theorem}\label{etafor}
If $x\in Otm(\mathfrak{G})$, then

\[\eta_{r}(\mathfrak{G})=\sum_{\delta=0}^{r}\sum_{J \in \mathcal{P} ([n])}(-1)^{|J|} \binom{n+\delta-\sum_{i\in J}m_{i} -1}{\delta-\sum_{i\in J}m_{i} }.\]
\end{theorem}

\begin{proof} Let $x\in Otm(\mathfrak{G})$, and $l_{i}(x):= max\{ x_{i}, m_{i}- (x_{i}+ 1)\}$ for all $i\in [n]$, then 
\begin{eqnarray*}
p(t,x)&=& \prod_{i=1}^{n} \left( \sum_{j=0}^{l_{i}(x)}|S_{j}(x_{i})|t^{j} \right)= \prod_{i=1}^{n} \left( \sum_{j=0}^{m_{i}-1}t^{j} \right)\\
      &=&\prod_{i=1}^{n} \left( \frac{1-t^{m_{i}}}{1-t} \right)=\left(\prod_{i=1}^{n}  (1-t^{m_{i}})  \right)\cdot \left( \frac{1}{1-t} \right)^{n}\\
      &=&\left(\sum_{J \in \mathcal{P} ([n])}(-1)^{|J|}t^{\sum_{i\in J}m_{i}}\right) \sum_{j=0}^{\infty} \binom{n+j-1}{j}t^{j}
\end{eqnarray*}

where the last equality is by Remark \ref{binomial} and \cite[Table 6.1]{tucker}. Thus, if one states that $\binom{a}{b}=0$ if $b<0$, for $0\leq \delta\leq r$ the coefficient of degree  $\delta$ of $p(t,x)$ is 
 \[p_{\delta}=\sum_{J \in \mathcal{P} ([n])}(-1)^{|J|} \binom{n+\delta-\sum_{i\in J}m_{i} -1}{\delta-\sum_{i\in J}m_{i} }.\]
 The rest follows from the fact that $\sum_{\delta=0}^{r}p_{\delta}=p(1,x)_{\leq r}=|B_{r}(x)|=\eta_{r}(\mathfrak{G})$ by Theorems \ref{local} and  \ref{fundamental}.
\end{proof}

Algorithm \ref{algoeta} provides the code in \textit{SageMath} to calculate $\eta_{r}(\mathfrak{G})$ .\\

\begin{remark}
To use the \textit{Python} functions given in Algorithms \ref{algoeta}, \ref{pdeltao}, \ref{pdeltae}, and \ref{pdelta}, we can copy them in a \texttt{.py} document that may be called later into a SageMath worksheet by writing \texttt{load("name of the .py document containing the functions")} at the beginning of the worksheet.    
\end{remark}

\begin{algorithm}[H]
\KwData{$\texttt{r}\in \mathbb{Z}_{\geq 0}$ and the \textit{Python} list $\texttt{m}=[m_{1},..., m_{n}]$ where $\mathfrak{G}=\prod_{i=1}^{n}[0,m_{i}-1]$.}
\KwResult{\texttt{eta(r, m)} computes $\eta_{r}(\mathfrak{G})$.}
\begin{verbatim}

def eta(r, m):
    if r==0:
        return 1
    
    n=len(m); indices = range(1, n+1); eta = 0
       
    for delta in range(r + 1):
        for J in Subsets(indices):
            sum_m_J = sum(m[i-1] for i in J)
            if delta - sum_m_J >= 0:
                binom_term = binomial(n + delta - sum_m_J - 1
                           , delta - sum_m_J)
                eta += (-1) ** len(J) * binom_term
    
    return eta
    
\end{verbatim}

\caption{A \textit{Python} function that applies Theorem \ref{etafor} to compute $\eta_{r}(\mathfrak{G})$ in \textit{SageMath}.}\label{algoeta}
\end{algorithm}

\begin{example}\label{exloc1}
Let $n=3$, $m_{1}=m_{2}=2$, $m_{3}=10$, and $r=5$. Then $\mathfrak{G}=\prod_{i=1}^{3}[0,m_{i}-1]=[0,1]\times [0,1]\times [0,9]$. By Algorithm \ref{algoeta},
\[\eta_{5}(\mathfrak{G})=\sum_{\delta=0}^{5}\sum_{J \in \mathcal{P} ([3])}(-1)^{|J|} \binom{3+\delta-\sum_{i\in J}m_{i} -1}{\delta-\sum_{i\in J}m_{i} }                      =20 \]
Let $m_{1}=m_{2}=7$, $\mathfrak{G}=[0,6]^{2}$ and $r=8$.  By Algorithm \ref{algoeta},
\[ \eta_{8}(\mathfrak{G})=\sum_{\delta=0}^{9}\sum_{J \in \mathcal{P} ([2])}(-1)^{|J|} \binom{2+\delta-\sum_{i\in J}m_{i} -1}{\delta-\sum_{i\in J}m_{i} }
                      = 39.\]


\begin{figure}[H]
\begin{tikzpicture}[scale=0.5, vertices/.style={draw,  circle, inner sep=2pt}]

\draw[fill=gray!10] (0,0) -- (0,6) -- (6,6) -- (6,0) -- cycle; 
\draw[step=1cm,black, very thin] (0,0) grid (6,6);

 [scale=1.3, vertices/.style={draw, fill=WildStrawberry!70, circle, inner sep=3pt}]

\node[ vertices,fill=Cerulean!60, label=below:{$0$}] at (0,0) {};
\node[vertices,fill=Cerulean!60, label=below:{$1$}] at (1,0) {};
\node[vertices,fill=Cerulean!60, label=below:{$2$}] at (2,0) {};
\node[vertices,fill=Cerulean!60, label=below:{$3$}] at (3,0) {};
\node[vertices,fill=Cerulean!60, label=below:{$4$}] at (4,0) {}; 
\node[vertices,fill=Cerulean!60, label=below:{$5$}] at (5,0) {};
\node[vertices,fill=Cerulean!60, label=below:{$6$}] at (6,0) {};

\node[vertices,fill=Cerulean!60, label=left:{$1$}] at (0,1) {};
\node[vertices,fill=Cerulean!60, label=left:{$2$}] at (0,2) {};
\node[vertices,fill=Cerulean!60, label=left:{$3$}] at (0,3) {};
\node[vertices,fill=Cerulean!60, label=left:{$4$}] at (0,4) {}; 
\node[vertices,fill=Cerulean!60, label=left:{$5$}] at (0,5) {};
\node[vertices,fill=Cerulean!60, label=left:{$6$}] at (0,6) {};

\node[vertices,fill=Cerulean!60] at (1,1) {};
\node[vertices,fill=Cerulean!60] at (2,1) {};
\node[vertices,fill=Cerulean!60] at (3,1) {};
\node[vertices,fill=Cerulean!60] at (4,1) {}; 
\node[vertices,fill=Cerulean!60] at (5,1) {};
\node[vertices,fill=Cerulean!60] at (6,1) {};
\node[vertices,fill=Cerulean!60] at (1,2) {};
\node[vertices,fill=Cerulean!60] at (2,2) {};
\node[vertices,fill=Cerulean!60] at (3,2) {};
\node[vertices,fill=Cerulean!60] at (4,2) {}; 
\node[vertices,fill=Cerulean!60] at (5,2) {};
\node[vertices,fill=Cerulean!60] at (6,2) {};

\node[vertices,fill=Cerulean!60] at (1,3) {};
\node[vertices,fill=Cerulean!60] at (2,3) {};
\node[vertices,fill=Cerulean!60] at (3,3) {};
\node[vertices,fill=Cerulean!60] at (4,3) {}; 
\node[vertices,fill=Cerulean!60] at (5,3) {};
\node[vertices,fill=Cerulean!60] at (1,4) {};
\node[vertices,fill=Cerulean!60] at (2,4) {};
\node[vertices,fill=Cerulean!60] at (3,4) {};
\node[vertices,fill=Cerulean!60] at (4,4) {}; 

\node[vertices,fill=Cerulean!60] at (1,5) {};
\node[vertices,fill=Cerulean!60] at (2,5) {};
\node[vertices,fill=Cerulean!60] at (3,5) {};

\node[vertices,fill=Cerulean!60] at (1,6) {};
\node[vertices,fill=Cerulean!60] at (2,6) {};

\end{tikzpicture}

\caption{ The $39$ elements represented by the cerulean dots in $\mathfrak{G}$ coincide with the computation of $\eta_{8}([0,6]\times[0,6])$ given by Algorithm \ref{algoeta}.} 
\label{fig6}
\end{figure}

\end{example}


\begin{theorem}\label{gammafor}
Let $x\in Inm(\mathfrak{G})$, $n\in \mathbb{Z}_{\geq 1}$, $\delta \in \mathbb{Z}_{>0}$, $p_{\delta}$ be the coefficient of degree $d$ of $p(t,x)$ and $\partial=deg(p(t,x))$. Let $l_{i}=\lceil \frac{m_{i}-1}{2} \rceil$ for all $i\in[n]$, $E=\{i\in [n]: 2\mid m_{i}\}$  and $\mathcal{P} ([n])^{*}=\mathcal{P} ([n])-\{\emptyset\}$. Then the following statements hold:

\begin{enumerate}

\item If $E= \emptyset$, then $p_{\partial}=2^{n}$ and for $0<\delta<\partial$
\[p_{\delta}= \sum_{J\in \mathcal{P}([n])^{*}}\sum_{A\in \mathcal{P}(J)} (-1)^{n-|J|+|A|}\cdot 2^{|J|}\cdot \binom{|J|+\delta-\sum_{i\in A}(l_{i}+1)-1}{\delta-\sum_{i\in A}(l_{i}+1)}.\]

\item If $E=[n]$, then $p_{\partial}=1$ and for $0<\delta<\partial$
\[ p_{\delta}= \sum_{J\in \mathcal{P}([n])^{*}}u_{J,\delta}+\sum_{A\in \mathcal{P}(J)^{*}} \sum_{B\in \mathcal{P}(A)} (-1)^{|J|-|A|+|B|}\cdot 2^{|A|}\cdot \binom{|A|+\delta-\sum_{i\in B \cup J^{c}}l_{i}-1}{\delta-\sum_{i\in B \cup J^{c}}l_{i}} \]
where $u_{J,\delta}=\begin{cases}
 (-1)^{|J|}  & \mbox{ if } \delta=\sum_{i\in J^{c}}l_{i}\\
0 & \text{ otherwise } 
\end{cases}.$\\

\item If $\emptyset\neq E \neq[n]$, $x_{e}:=(x_{i})_{i\in E}$, $x_{o}:=(x_{i})_{i\notin E}$, $q_{j}$ and $h_{i}$ denote the coefficients of degree $j$ and $i$ of $q(t,x_{e})$ and $h(t,x_{o})$, respectively, then $p_{\delta}=\sum_{j=0}^{deg(q)}q_{j} h_{\delta-j}$.

\end{enumerate}

\end{theorem}

\begin{proof}
\begin{enumerate}
\item Let $E=\emptyset$, i.e., $2\nmid m_{i}$ for all $i\in[n]$. Then, by Theorem \ref{local}
\begin{eqnarray*}
p(t,x)&=& \prod_{i=1}^{n} \left( \sum_{j=0}^{l_{i}(x)}|S_{j}(x_{i})|t^{j} \right)= \prod_{i=1}^{n} \left[ 2 \left(\sum_{j=0}^{l_{i}}t^{j}\right)-1 \right]\\
      &=&\prod_{i=1}^{n} \left[ 2 \left(\frac{1-t^{l_{i}+1}}{1-t} \right)-1 \right]=\sum_{J\in \mathcal{P}([n])}(-1)^{n-|J|}\cdot 2^{|J|}\cdot \prod_{i\in J}\left( \frac{1-t^{l_{i}+1}}{1-t}\right)\\
      &=&\sum_{J\in \mathcal{P}([n])}(-1)^{n-|J|}\cdot 2^{|J|}\cdot \prod_{i\in J}\left( 1-t^{l_{i}+1}\right)\cdot \left( \frac{1}{1-t}\right)^{|J|}\\
      &=&(-1)^{n}+\sum_{J\in \mathcal{P}([n])^{*}}(-1)^{n-|J|}\cdot 2^{|J|}\cdot \prod_{i\in J}\left( 1-t^{l_{i}+1}\right)\cdot \left( \frac{1}{1-t}\right)^{|J|}\\
      &=&(-1)^{n}+ \sum_{J\in \mathcal{P}([n])^{*}}(-1)^{n-|J|}\cdot 2^{|J|}\cdot \left[\sum_{A\in \mathcal{P}(J)} (-1)^{|A|}t^{\sum_{i\in A}(l_{i}+1)}\right]\\
      &\cdot & \left[ \sum_{j=0}^{\infty}\binom{|J|+j-1}{j}t^{j}\right]
\end{eqnarray*}
where the fourth and sixth equalities are by Remark \ref{binomial}. Thus  $p_{\partial}=2^{n}$, and if one states  that $\binom{a}{b}=0$ if $b<0$,  the coefficient of degree $0<\delta< \sum_{i=1}^{n}l_{i}=\partial$ of  $p(t,x)$ is given by 
\[p_{\delta}= \sum_{J\in \mathcal{P}([n])^{*}}\sum_{A\in \mathcal{P}(J)} (-1)^{n-|J|+|A|}\cdot 2^{|J|}\cdot \binom{|J|+\delta-\sum_{i\in A}(l_{i}+1)-1}{\delta-\sum_{i\in A}(l_{i}+1)}.\]
\item Let $E=[n]$, i.e., $2\mid m_{i}$ for all $i\in[n]$. Then, by Theorem \ref{local}
\begin{eqnarray*}
p(t,x)&=& \prod_{i=1}^{n} \left( \sum_{j=0}^{l_{i}(x)}|S_{j}(x_{i})|t^{j} \right)= \prod_{i=1}^{n} \left[ 1+ \left(\sum_{j=1}^{l_{i}-1}2t^{j}\right)+t^{l_{i}} \right]\\
      &=& \prod_{i=1}^{n} \left[ \left(2\left( \sum_{j=0}^{l_{i}-1}t^{j}\right)-1\right) + t^{l_{i}} \right]\\
      &=& \sum_{J\in \mathcal{P}([n])}    \prod_{i\in J} \left(2\left( \sum_{j=0}^{l_{i}-1}t^{j}\right)-1\right)   \cdot  \prod_{i\notin J}  t^{l_{i}}   \\
      &=& \sum_{J\in \mathcal{P}([n])}    \left[ \sum_{A\in \mathcal{P}(J)} (-1)^{|J|-|A|}\cdot 2^{|A|} \cdot \prod_{i\in A}\left( \frac{1-t^{l_{i}}}{1-t} \right) \right]   \cdot     t^{\sum_{i\notin J}l_{i}}   \\ 
      \end{eqnarray*}
      
\begin{eqnarray*}      
      &=& t^{\sum_{i=1}^{n}l_{i}}+ \sum_{J\in \mathcal{P}([n])^{*}}  \left[ (-1)^{|J|} +   \sum_{A\in \mathcal{P}(J)^{*}} (-1)^{|J|-|A|}\cdot 2^{|A|} \cdot \prod_{i\in A}\left( \frac{1-t^{l_{i}}}{1-t} \right) \right] \\
      &\cdot &     t^{\sum_{i\notin J}l_{i}} \\
      &=&   t^{\sum_{i=1}^{n}l_{i}}+ \sum_{J\in \mathcal{P}([n])^{*}}  \left[ (-1)^{|J|} +   \sum_{A\in \mathcal{P}(J)^{*}} (-1)^{|J|-|A|}\cdot 2^{|A|}  \right.\\
      & \textcolor{black}{\cdot} &   \left.  \left[\sum_{B\in \mathcal{P}(A)} (-1)^{|B|}t^{\sum_{i\in B}l_{i}}\right] \left[ \sum_{j=0}^{\infty}\binom{|A|+j-1}{j}t^{j} \right] \right]  \cdot t^{\sum_{i\in J^{c}}l_{i}}  \\
      &=& t^{\sum_{i=1}^{n}l_{i}}+ \sum_{J\in \mathcal{P}([n])^{*}} \left[ (-1)^{|J|}  t^{\sum_{i\in J^{c}}l_{i}} +   \sum_{A\in \mathcal{P}(J)^{*}} (-1)^{|J|-|A|}\cdot 2^{|A|}  \right.\\
      &\textcolor{black}{\cdot} &   \left.  \left[\sum_{B\in \mathcal{P}(A)} (-1)^{|B|}t^{\sum_{i\in B \cup J^{c}}l_{i}}\right] \left[ \sum_{j=0}^{\infty}\binom{|A|+j-1}{j}t^{j} \right] \right] 
\end{eqnarray*}

where the fourth, fifth, and seventh equalities are by Remark \ref{binomial}. Thus  $p_{\partial}=1$, and if one states that $\binom{a}{b}=0$  if $b<0$,  the coefficient of degree $0<\delta<\partial$ of  $p(t,x)$ is given by 
\[ p_{\delta}=\sum_{J\in \mathcal{P}([n])^{*}}u_{J,\delta}+\sum_{A\in \mathcal{P}(J)^{*}} \sum_{B\in \mathcal{P}(A)} (-1)^{|J|-|A|+|B|} \cdot 2^{|A|}\binom{|A|+\delta-\sum_{i\in B \cup J^{c}}l_{i}-1}{\delta-\sum_{i\in B \cup J^{c}}l_{i}}\]
where $u_{J,\delta}=\begin{cases}
 (-1)^{|J|}  & \mbox{ if } \delta=\sum_{i\in J^{c}}l_{i}\\
0 & \text{ otherwise } 
\end{cases}.$\\

\item Suppose $\emptyset\neq E \neq[n]$, $x_{e}=(x_{i})_{i\in E}$ and $x_{o}=(x_{i})_{i\notin E}$. Let $q_{j}$ and $h_{i}$ be the coefficients of degree $j$ and $i$ of $q(t,x_{e})$ and $h(t,x_{o})$, respectively. Then $p(t,x)=q(t,x_{e})h(t,x_{o})$ and so $p_{\delta}=\sum_{j=0}^{deg(q)}q_{j} h_{\delta-j}$.

\end{enumerate}

\end{proof}

The function in Algorithm \ref{pdelta} uses the functions given in Algorithms \ref{pdeltao}, \ref{pdeltae} to provide the \textit{SageMath} code to calculate $p_{\delta}$  for $\mathfrak{G}=\prod_{i=1}^{n}[0,m_{i}-1]$ when $E$ is arbitrary.\\

\begin{algorithm}[H]
\KwData{$\texttt{delta}\in \mathbb{Z}_{\geq 0}$ and the \textit{Python} list $\texttt{m}=[m_{1},..., m_{n}]$ where $\mathfrak{G}=\prod_{i=1}^{n}[0,m_{i}-1]$.}
\KwResult{ \texttt{p\_delta\_o(delta, m)} computes the coefficient $p_{\delta}$ of degree $\delta$ of $p(t,\text{any innermost point})$ when $E=\emptyset$.}
\begin{verbatim}

def p_delta_o(delta, m):
    if delta==0:
        return 1
    
    n=len(m); l = [ceil((mi - 1) / 2) for mi in m]
    indices=range(1, n+1)
    E = {i for i in indices if m[i-1] % 2 == 0}
    partial=sum(li for li in l) 
    p_delta = 0
        
    if len(E) != 0:
        raise ValueError("The m_i's must be odd")
    
    elif delta == partial:
        return 2**n
    
    P_n= Subsets(indices)
    for J in P_n:
        if len(J) > 0:
            for A in Subsets(J):
                sum_l_A = sum(l[i-1] + 1 for i in A)
                if delta - sum_l_A >= 0:
                    binom_term = binomial(len(J) + delta - sum_l_A
                               - 1, delta - sum_l_A)
                    p_delta += (-1)**(n - len(J) + len(A))
                             * 2**len(J) * binom_term

    return p_delta
    
\end{verbatim}
\caption{An algorithm that applies Theorem \ref{gammafor} (part $1$) to compute $p_{\delta}$ in \textit{SageMath}.}\label{pdeltao}
\end{algorithm}

\begin{algorithm}[H]
\KwData{$\texttt{delta}\in \mathbb{Z}_{\geq 0}$ and the \textit{Python} list $\texttt{m}=[m_{1},..., m_{n}]$ where $\mathfrak{G}=\prod_{i=1}^{n}[0,m_{i}-1]$.}
\KwResult{ \texttt{p\_delta\_o(delta, m)} computes the coefficient $p_{\delta}$ of degree $\delta$ of $p(t,\text{any innermost point})$ when $E=[n]$.}
\begin{verbatim}

def p_delta_e(delta, m):
    if delta==0:
        return 1
    
    n=len(m); l = [ceil((mi - 1) / 2) for mi in m]
    indices=range(1, n+1)
    E = {i for i in indices if m[i-1] % 2 == 0}
    partial=sum(li for li in l) 
    p_delta = 0
        
    if len(E) != n:
        raise ValueError("The m_i's must be even")

    if delta == partial:
        return 1
    
    P_n = Subsets(indices)
    for J in P_n:
        if len(J) > 0:  # Ensure J is non-empty
            J_c = set(indices) - set(J)
            u_J_delta = (-1)**(len(J)) if delta ==
                        sum(l[i-1] for i in J_c) else 0
            
            term_sum = u_J_delta
            
            for A in Subsets(J):
                if len(A) > 0:
                    for B in Subsets(A):
                        sum_l_B_J_c = sum(l[i-1] for i in B) 
                                    + sum(l[i-1] for i in J_c)
                        if delta - sum_l_B_J_c >= 0:
                            binom_term = binomial(len(A) + delta 
                                       - sum_l_B_J_c - 1, delta 
                                       - sum_l_B_J_c)
                            term_sum += (-1)**(len(J) - len(A) 
                                      + len(B))  *  2**len(A)  
                                      *  binom_term
            
            p_delta += term_sum
    
    return p_delta
    
\end{verbatim}
\caption{An algorithm that applies Theorem \ref{gammafor} (part $2$) to compute $p_{\delta}$ in \textit{SageMath}.}\label{pdeltae}
\end{algorithm}

\begin{algorithm}[H]
\KwData{$\texttt{delta}\in \mathbb{Z}_{\geq 0}$ and the \textit{Python} list $\texttt{m}=[m_{1},..., m_{n}]$ where $\mathfrak{G}=\prod_{i=1}^{n}[0,m_{i}-1]$.}
\KwResult{ \texttt{p\_delta(delta, m)} computes the coefficient $p_{\delta}$ of degree $\delta$ of $p(t,\text{any innermost point})$ when $E$ is arbitrary.}
\begin{verbatim}

def p_delta(delta, m):
    if delta==0:
        return 1
    
    n=len(m); indices=range(1, n+1)
    m_e=[m[i-1] for i in indices if m[i-1] % 2 == 0]
    partial_e=sum(ceil((m - 1) / 2) for m in m_e)
    m_o=[m[i-1] for i in indices if m[i-1] % 2 != 0]
    p_delta = 0
    
    return sum(p_delta_e(j, m_e) *  p_delta_o(delta-j, m_o)
               for j in range(partial_e + 1))
               
\end{verbatim}

\caption{An algorithm that applies Theorem \ref{gammafor} (part $3$) to compute $p_{\delta}$ in \textit{SageMath}. It works even when $E=\emptyset$ or $E=[n]$, because \texttt{p\_delta\_e(0, [])=p\_delta\_o(0, [])=1}.}\label{pdelta}
\end{algorithm}

\vspace{0.5cm}

\begin{algorithm}[H]
\KwData{$\texttt{r}\in \mathbb{Z}_{\geq 0}$ and the \textit{Python} list $\texttt{m}=[m_{1},..., m_{n}]$ where $\mathfrak{G}=\prod_{i=1}^{n}[0,m_{i}-1]$.}
\KwResult{ \texttt{gamma(r, m)} computes $\gamma_{r}(\mathfrak{G})$.}
\begin{verbatim}

def gamma(r, m):
    if r==0:
        return 1
    
    return 1 + sum(p_delta(delta, m) for delta in range(1, r+1))
    
\end{verbatim}

\caption{Since $\gamma_{r}(\mathfrak{G})=p(1, x)_{\leq r}=1+\sum_{\delta=1}^{r}p_{\delta}$ (by Theorems \ref{local} and  \ref{fundamental}) and $p_{\delta}=\texttt{p\_delta(delta, m)}$ (where  \texttt{p\_delta(delta, m)} is given as in Algorithm \ref{pdelta}), this function computes $\gamma_{r}(\mathfrak{G})$.}\label{algogamma}
\end{algorithm}

\vspace{0.5cm}

In the trivial cases described in Remark \ref{trivial} it is not necessary to use the formulas developed in Theorems \ref{etafor} nor \ref{gammafor}.

\begin{remark}\label{trivial}
Let $x\in Otm(\mathfrak{G})$ and $y\in Inm(\mathfrak{G})$. Let $\mathbf{B}_{r}(x)$ and $ \mathbf{B}_{r}(y)$ denote the $r$-ball centred at $x$ and $y$ in $\mathbb{Z}^{n}$, respectively, then

\begin{enumerate}

\item $\mathbf{B}_{r}(x)\subseteq \mathfrak{G}$ if and only if $r\leq \min \{ m_{i}-1\}_{i=1}^{n}$. In particular, if $r\leq \min \{ m_{i}-1\}_{i=1}^{n}$, then $\eta_{r}(\mathfrak{G})=\sum_{j=0}^{r}\binom{n+j-1}{j}$. \\

\item $ \mathfrak{G}\subseteq \mathbf{B}_{r}(x)$ if and only if $ \sum_{i=1}^{n} (m_{i}-1)\leq r$. In particular, if $\sum_{i=1}^{n} (m_{i}-1)\leq r$, then $\eta_{r}(\mathfrak{G})=\prod_{i=1}^{n}m_{i}$.\\

\item $\mathbf{B}_{r}(y)\subseteq \mathfrak{G}$ if and only if $r\leq \min \{ \lfloor \frac{m_{i}-1}{2} \rfloor\}_{i=1}^{n}$. In particular, if $r\leq \min \{ \lfloor \frac{m_{i}-1}{2} \rfloor\}_{i=1}^{n}$, then $\gamma_{r}(\mathfrak{G})=\sum_{j=0}^{\min\{r,n\}}2^{j}\binom{n}{j}\binom{r}{j}$. \\

\item $\mathfrak{G} \subseteq \mathbf{B}_{r}(y)$ if and only if $\sum_{i=1}^{n} \lceil \frac{m_{i}-1}{2} \rceil \leq r $. In particular, if $\sum_{i=1}^{n} \lceil \frac{m_{i}-1}{2} \rceil \leq r $, then $\gamma_{r}(\mathfrak{G})=\prod_{i=1}^{n}m_{i}$.

\end{enumerate}

\end{remark}

\begin{remark}\label{eqsball}
Let $r\in \mathbb{Z}_{>0}$, $m_{i},n_{i}\in \mathbb{Z}_{\geq 2}$, $\mathfrak{G}_{0}:=\prod_{i=1}^{n}[0,m_{i}-1]$, $ \mathfrak{G}_{1}:=\prod_{i=1}^{n}[0,n_{i}-1]$. If $x\in\mathfrak{G}_{0}$ and $y\in\mathfrak{G}_{1}$ are such that $|S_{j}(x_{i})|=|S_{j}(y_{i})|$ for all $i\in [n]$ and $0\leq j\leq r$, then $|B_{r}(x)|=|B_{r}(y)|$. 
\end{remark}

\begin{proof}

Let $l_{i}(x):=max\{ x_{i}, (m_{i}-1)-x_{i}\}$,  and $l_{i}(y):=max\{ y_{i}, (m_{i}-1)-y_{i}\}$ for all $i\in [n]$. If $i\in [n]$ is such that $l_{i}(x)=l_{i}(y)$, then $min\{r, l_{i}(x)\}= min\{r, l_{i}(y)\}$. On the other hand, if $i\in [n]$ is such that $l_{i}(x)<l_{i}(y)$, then $r\leq l_{i}(x)<l_{i}(y)$ or $l_{i}(x)<r\leq l_{i}(y)$ or $l_{i}(x)< l_{i}(y) < r$. It will be shown that the only possible case is $r\leq l_{i}(x)<l_{i}(y)$. If $l_{i}(x)<r\leq l_{i}(y)$, then $0<|S_{l_{i}(x)+1}(y_{i})|=|S_{l_{i}(x)+1}(x_{i})|=0$. 
If $l_{i}(x)< l_{i}(y) < r$, then $0<|S_{l_{i}(y)}(y_{i})|=|S_{l_{i}(y)}(x_{i})|=0$. 
In conclusion, if $l_{i}(x)<l_{i}(y)$, then $r\leq l_{i}(x)<l_{i}(y)$ and so $r=min\{r, l_{i}(x)\}=min\{r, l_{i}(y)\}$. A symmetric argument applies if $i\in[n]$ is such that $l_{i}(y)<l_{i}(x)$. Hence, $min\{r, l_{i}(x)\}=min\{r, l_{i}(y)\}$ for all $i\in[n]$. Thus, if $s\leq r$, the coefficient of degree $s$ of $p(t,x)$ or $q(t,y)$ is equal to the coefficient of degree $s$ of the polynomial $\prod_{i=1}^{n}\left( \sum_{j=0}^{min\{r, l_{i}(x)\}} |S_{j}(x_{i})|t^{j} \right) $, which is equal to $\prod_{i=1}^{n}\left( \sum_{j=0}^{min\{r, l_{i}(y)\}} |S_{j}(y_{i})|t^{j} \right)$ because $|S_{j}(x_{i})|=|S_{j}(y_{i})|$ for all $i\in [n]$ and $0\leq j\leq r$. Therefore $|B_{r}(x)|=p(1,x)_{\leq r}=q(1,y)_{\leq r}=|B_{r}(y)|$ by Theorem \ref{local}. \\
  
\end{proof}

\begin{example}\label{gammaex}  

If $m_{1}=m_{2}=4$ and $m_{3}\geq 9$. Then $\mathfrak{G}=[0,3]^{2}\times [0, m_{3}-1]$.  Let $\mathfrak{G}_{0}=[0,3]^{2}\times [0, 8]$, by Remark \ref{eqsball} $\gamma_{4}(\mathfrak{G})=\gamma_{4}(\mathfrak{G}_{0})$ and $\eta_{4}(\mathfrak{G})=\eta_{4}(\mathfrak{G}_{0})$. Thus $\gamma_{4}(\mathfrak{G})=80$ and $\eta_{2}(\mathfrak{G})=10$ (by Algorithms \ref{algoeta} and \ref{algogamma}). Therefore  \[\dfrac{\prod_{i=1}^{3}m_{i}}{\gamma_{4}(\mathfrak{G})}=\dfrac{16m_{3}}{80}=\dfrac{m_{3}}{5} \leq \mathcal{A}_{\mathfrak{G}}(3,5)\leq  \dfrac{\prod_{i=1}^{3}m_{i}}{\eta_{2}(\mathfrak{G})}=\dfrac{16m_{3}}{10}=\dfrac{8m_{3}}{5},\]
by Theorems \ref{spb} and \ref{gvb}.

\end{example}

Tables \ref{grids2d} and \ref{grids3d} provide the intervals where $A_{\mathfrak{G}}(n,d) $ belongs,  with $\mathfrak{G}$ being the set of the vertices of several Rubik's cubes and their faces and $d=1,...,10$.

\begin{table}[H]
\centering
\begin{tabular}{|c|c|c|c|c|c|c|c|c|}
 \hline
 \diagbox{$d$}{\tiny{$A_{\mathfrak{G}}(2,d)\in$}}{$\mathfrak{G}$} & $[0, 3]^{2}$ & $[0, 4]^{2}$ & $[0, 5]^{2}$ & $[0, 6]^{2}$ & $[0, 7]^{2}$ & $[0, 8]^{2}$ & $[0, 9]^{2}$ & $[0, 10]^{2}$ \\ 
 \hline
 $1$ & $\left[9, 9\right]$ & $\left[16, 16\right]$ & $\left[25, 25\right]$ & $\left[36, 36\right]$ & $\left[49, 49\right]$ & $\left[64, 64\right]$ & $\left[81, 81\right]$ & $\left[100, 100\right]$ \\ 
 \hline
 $2$ & $\left[2, 9\right]$ & $\left[4, 16\right]$ & $\left[5, 25\right]$ & $\left[8, 36\right]$ & $\left[10, 49\right]$ & $\left[13, 64\right]$ & $\left[17, 81\right]$ & $\left[20, 100\right]$ \\ 
 \hline
 $3$ & $\left[1, 3\right]$ & $\left[2, 5\right]$ & $\left[2, 8\right]$ & $\left[3, 12\right]$ & $\left[4, 16\right]$ & $\left[5, 21\right]$ & $\left[7, 27\right]$ & $\left[8, 33\right]$ \\ 
 \hline
 $4$ & $\left[1, 3\right]$ & $\left[2, 5\right]$ & $\left[2, 8\right]$ & $\left[2, 12\right]$ & $\left[2, 16\right]$ & $\left[3, 21\right]$ & $\left[4, 27\right]$ & $\left[4, 33\right]$ \\ 
 \hline
 $5$ & $\left[1, 1\right]$ & $\left[1, 2\right]$ & $\left[1, 4\right]$ & $\left[2, 6\right]$ & $\left[2, 8\right]$ & $\left[2, 10\right]$ & $\left[2, 13\right]$ & $\left[3, 16\right]$ \\ 
 \hline
 $6$ & $\left[1, 1\right]$ & $\left[1, 2\right]$ & $\left[1, 4\right]$ & $\left[2, 6\right]$ & $\left[2, 8\right]$ & $\left[2, 10\right]$ & $\left[2, 13\right]$ & $\left[2, 16\right]$ \\ 
 \hline
 $7$ & $\left[1, 1\right]$ & $\left[1, 1\right]$ & $\left[1, 2\right]$ & $\left[1, 3\right]$ & $\left[1, 4\right]$ & $\left[2, 6\right]$ & $\left[2, 8\right]$ & $\left[2, 10\right]$ \\ 
 \hline
 $8$ & $\left[1, 1\right]$ & $\left[1, 1\right]$ & $\left[1, 2\right]$ & $\left[1, 3\right]$ & $\left[1, 4\right]$ & $\left[2, 6\right]$ & $\left[2, 8\right]$ & $\left[2, 10\right]$ \\ 
 \hline
 $9$ & $\left[1, 1\right]$ & $\left[1, 1\right]$ & $\left[1, 1\right]$ & $\left[1, 2\right]$ & $\left[1, 3\right]$ & $\left[1, 4\right]$ & $\left[1, 5\right]$ & $\left[2, 6\right]$ \\ 
 \hline
 $10$ & $\left[1, 1\right]$ & $\left[1, 1\right]$ & $\left[1, 1\right]$ & $\left[1, 2\right]$ & $\left[1, 3\right]$ & $\left[1, 4\right]$ & $\left[1, 5\right]$ & $\left[2, 6\right]$ \\ 
 \hline
\end{tabular}
\caption{By the Hamming and Gilbert bounds (given in Theorems \ref{spb} and \ref{gvb}),
$\mathcal{A}_{\mathfrak{G}}(2,d) \in \left[ \frac{|\mathfrak{G}|}{\gamma_{d-1}(\mathfrak{G})}, \frac{|\mathfrak{G}|}{\eta_{\floor{\frac{d-1}{2}}}(\mathfrak{G})}\right]$. The extremes of the intervals were adjusted by taking their ceil (for the lower ends) and floor (for the upper ends).}
\label{grids2d}
\end{table}

\begin{table}[H]
\centering
\begin{tabular}{|c|c|c|c|c|c|c|c|c|}
 \hline
 \diagbox{$d$}{\tiny{$A_{\mathfrak{G}}(3,d)\in $}}{$\mathfrak{G}$} & $[0, 3]^{3}$ & $[0, 4]^{3}$ & $[0, 5]^{3}$ & $[0, 6]^{3}$ & $[0, 7]^{3}$ & $[0, 8]^{3}$ & $[0, 9]^{3}$ & $[0, 10]^{3}$ \\ 
 \hline
 1 & $\left[3^3, 3^3\right]$ & $\left[4^3, 4^3\right]$ & $\left[5^3, 5^3\right]$ & $\left[6^3, 6^3\right]$ & $\left[7^3, 7^3\right]$ & $\left[8^3, 8^3\right]$ & $\left[9^3, 9^3\right]$ & $\left[10^3, 10^3\right]$ \\ 
 \hline
 2 & $\left[4, 3^3\right]$ & $\left[10, 4^3\right]$ & $\left[18, 5^3\right]$ & $\left[31, 6^3\right]$ & $\left[49, 7^3\right]$ & $\left[74, 8^3\right]$ & $\left[105, 9^3\right]$ & $\left[143, 10^3\right]$ \\ 
 \hline
 3 & $\left[2, 6\right]$ & $\left[3, 16\right]$ & $\left[5, 31\right]$ & $\left[9, 54\right]$ & $\left[14, 85\right]$ & $\left[21, 128\right]$ & $\left[30, 182\right]$ & $\left[40, 250\right]$ \\ 
 \hline
 4 & $\left[1, 6\right]$ & $\left[2, 16\right]$ & $\left[3, 31\right]$ & $\left[4, 54\right]$ & $\left[6, 85\right]$ & $\left[9, 128\right]$ & $\left[12, 182\right]$ & $\left[16, 250\right]$ \\ 
 \hline
 5 & $\left[1, 2\right]$ & $\left[2, 6\right]$ & $\left[2, 12\right]$ & $\left[2, 21\right]$ & $\left[3, 34\right]$ & $\left[5, 51\right]$ & $\left[6, 72\right]$ & $\left[8, 100\right]$ \\ 
 \hline
 6 & $\left[1, 2\right]$ & $\left[2, 6\right]$ & $\left[2, 12\right]$ & $\left[2, 21\right]$ & $\left[2, 34\right]$ & $\left[3, 51\right]$ & $\left[4, 72\right]$ & $\left[5, 100\right]$ \\ 
 \hline
 7 & $\left[1, 1\right]$ & $\left[1, 3\right]$ & $\left[1, 6\right]$ & $\left[2, 10\right]$ & $\left[2, 17\right]$ & $\left[2, 25\right]$ & $\left[3, 36\right]$ & $\left[3, 50\right]$ \\ 
 \hline
 8 & $\left[1, 1\right]$ & $\left[1, 3\right]$ & $\left[1, 6\right]$ & $\left[2, 10\right]$ & $\left[2, 17\right]$ & $\left[2, 25\right]$ & $\left[2, 36\right]$ & $\left[2, 50\right]$ \\ 
 \hline
 9 & $\left[1, 1\right]$ & $\left[1, 2\right]$ & $\left[1, 3\right]$ & $\left[1, 5\right]$ & $\left[2, 9\right]$ & $\left[2, 14\right]$ & $\left[2, 20\right]$ & $\left[2, 28\right]$ \\ 
 \hline
 10 & $\left[1, 1\right]$ & $\left[1, 2\right]$ & $\left[1, 3\right]$ & $\left[1, 6\right]$ & $\left[1, 9\right]$ & $\left[2, 14\right]$ & $\left[2, 20\right]$ & $\left[2, 28\right]$ \\ 
 \hline
\end{tabular}
\caption{By the Hamming and Gilbert bounds (given in Theorems \ref{spb} and \ref{gvb}),
$\mathcal{A}_{\mathfrak{G}}(3,d) \in \left[ \frac{|\mathfrak{G}|}{\gamma_{d-1}(\mathfrak{G})}, \frac{|\mathfrak{G}|}{\eta_{\floor{\frac{d-1}{2}}}(\mathfrak{G})}\right]$. The extremes of the intervals were adjusted by taking their ceil (for the lower ends) and floor (for the upper ends).}
\label{grids3d}
\end{table}

The lower ends in Tables \ref{grids2d} and \ref{grids3d} are, especially, lower bounds for the sizes of $(d-1)$-covering codes.\\

\section{Cyclic codes}\label{S6}

A code $C$ of length $n$ over a finite group $G$ is a group code if $C$ is a subgroup of $G^n$ (the term ``group code'' is also used to refer to an ideal of a group algebra of a finite group over a finite field, see for example \cite{garcia-tapia}).  In  \cite{Elias} it was proved that group codes defined over the additive structure of the vector space $\F_{2}^{n}$ achieve Shannon's channel capacity for symmetric channels. Later, in \cite{Ahlswede} it was shown that this is not true for codes over $\F_{q}$ when $q\neq 2$. Group codes (with the Hamming distance) over more general finite groups were widely studied by several authors in the $1990$'s. For instance, in \cite{Forney92} it was proved that a group code over a general group $G$ cannot have better parameters than a conventional linear code over a field of the same size of $G$ (see \cite{Forney92, Forney93, Massey, Biglieri, Interlando} for other results on group codes of that decade). Some more recent works involving group codes have been in the theory of discrete memoryless channels or in point-to-point communication systems, see for example  \cite{Krithivasan, Sahebi1, Sahebi2}.\\

In this section, it will be presented a way of computing the minimum Hamming distance for codes that are cyclic subgroups of an abelian group, and some lower bounds for their minimum Manhattan distance in terms of their minimum Hamming and Lee distances.\\

 Let $C_{m_{i}}$ be the cyclic group of order $m_{i}$ for all  $ i\in[n]$, and $G$ be an abelian finite group isomorphic to the external direct product $C_{m_{1}}\times \cdots \times C_{m_{n}}$. Then, there exists a generating set $\{g_{i}\}_{i=1}^{n}$ of $G$ with $o(g_{i})=m_{i}$ such that $G$ is the internal direct product of the groups $\langle g_{i} \rangle's$. Thus any $g\in G$ can be uniquely represented as $g=\prod_{i=1}^n g_i^{\epsilon_{i}}$ with $\epsilon_{i}\in [0,m_{i}-1]$ $\forall i\in[n]$. Hence the function $\upsilon:G \rightarrow  \mathfrak{G}$ given by $\upsilon(\prod_{i=1}^n g_i^{\epsilon_{i}})=(\epsilon_{i})_{i=1}^n$ is a bijection. In this way, the Manhattan distance can be naturally defined in $G$ as
\[
d\left(\prod_{i=1}^n g_i^{\epsilon_i}, \prod_{i=1}^n g_i^{\delta_i}\right) := d\left(\upsilon\left(\prod_{i=1}^n g_i^{\epsilon_{i}}\right), \upsilon\left(\prod_{i=1}^n g_i^{\delta_{i}}\right)\right)=d((\epsilon_i)_{i=1}^{n}, (\delta_{i})_{i=1}^{n})= \sum_{i=1}^n |\epsilon_{i}-\delta_{i}|.
\]
This permits to interpret $G$ as an $n$-dimensional grid. Similarly, the Lee distance can be defined in $G$ as 
\[d_{L}\left(\prod_{i=1}^n g_i^{\epsilon_{i}}, \prod_{i=1}^n g_i^{\delta_{i}}\right):=d_{L}\left(\upsilon\left(\prod_{i=1}^n g_i^{\epsilon_{i}}\right), \upsilon\left(\prod_{i=1}^n g_i^{\delta_{i}}\right)\right)=\sum_{i=1}^{n}\min \left\{|\epsilon_{i}-\delta_{i}|, o(g_{i})-|\epsilon_{i}-\delta_{i}|\right\}.\]
Note that, for the Manhattan and Lee distances to be well defined as functions, the representation of the elements of $G$ must always be in terms of the generating set $\{g_{i}\}_{i=1}^{n}$ and the exponent joining $g_{i}$ must belong to the interval $[0,m_{i}-1]$ $\forall i\in[n]$. The Manhattan and Lee distances could change if the generating set is changed. For example, if $G=\langle g \rangle$ with $o(g)=9$, then the Manhattan distance between $g$ and $g^{5}$ with respect to the generating set $\{g\}$ is $d(g,g^{5})=|1-5|=4$, but with respect to the generating set  $\{g^{4}\}$ is $d(g,g^{5})=d((g^{4})^{7},(g^{4})^{8})=|7-8|=1$, because $g=(g^{4})^{7}$ and $g^{5}=(g^{4})^{8}$. Similarly, the Lee distance between $g$ and $g^{2}$ with respect to the generating set $\{g\}$ is $d_{L}(g,g^{2})=|1-2|=1$, but with respect to the generating set  $\{g^{4}\}$ is $d_{L}(g,g^{2})=d_{L}((g^{4})^{7},(g^{4})^{5})=|7-5|=2$, because  $g^{2}=(g^{4})^{5}$.\\ 

If a grid code is also a subgroup of $G$, it will be said that it is a group code. A cyclic code in $G$ will be a group code that is a cyclic subgroup of $G$ (the term ``cyclic code'' is also used to refer to an ideal of the ring $\frac{\mathbb{F}_{q}[x]}{\left\langle x^{n}-1\right\rangle}$, see for example \cite[Chapter 4]{Huffman}). If $g=\prod_{i=1}^{n}g_{i}^{u_{i}}\in G$ with $0\leq u_{i} \leq m_{i}-1$, its support with respect to the generating set $\{g_{i}\}_{i=1}^{n}$ of $G$ will be the set $Supp(g):=\{i\in[n]: u_{i}\neq 0\}$. Then the Hamming distance will be $d_{H}(g,h):=|Supp(gh^{-1})|$ $\forall g, h \in G$.\\

Not all the generating sets of $G$ are suitable for defining the Hamming, Lee, and Manhattan distances. For example, If $G$ is isomorphic to $C_{9}\times C_{2}$, then there exists a generating set $\{g_{1},g_{2}\}$ of $G$ with $o(g_{1})=9$ and $o(g_{2})=2$. Let $h_{1}:=g_{1}g_{2}$ then $S:=\{h_{1}, g_{2}\}$ is a generating set for $G$. However, $g_{1}=(h_{1})^{1}(g_{2})^{1}=(h_{1})^{10}(g_{2})^{0}$, so $g_{1}$ has two distinct representations as product of powers of the generators in $S$, where the exponents joining $h_{1}$ belong to $[0, o(h_{1})-1]=[0,17]$ and the exponents joining $g_{2}$ belong to $[0, o(g_{1})-1]=[0,1]$. Thus if one wanted to define the Hamming, Lee, and Manhattan distances in $G$ with respect to $S$ (in the sense previously presented),  then $d_{H}(g_{1},1)$, $d_{L}(g_{1},1)$ or $d(g_{1},1)$ would not be well defined.\\


 For a given grid code $\codeC$ in $G$, $\Delta_{H}(\codeC):=\max \{d_{H}(g,h): g,h\in \codeC \wedge \, g\neq h \}$ and $\Delta(\codeC):=\max \{d(g,h): g,h\in \codeC \wedge \, g\neq h \}$. Note that $\Delta_{H}(G)=n$ and $\Delta(G)=\sum_{i=1}^{n}(m_{i}-1)$.\\

\begin{theorem}\label{cyclic-code} Let $g_{0}=\prod_{i=1}^{n} g_{i}^{e_{i}}\neq 1$ with $0\leq e_{i}\leq m_{i}-1$ $\forall i \in[n]$, $\codeC= \left\langle g_{0} \right\rangle$ and $S=Supp(g_{0})$. Let $l_{i}=gcd(e_{i}, m_{i})$ and  $c_{i}=\frac{e_{i}}{l_{i}}$ for all $i\in S$, and $l=\min \{l_{i}\}_{i\in S}$. Let $\mathcal{O}_{J}=lcm\{\frac{m_{i}}{l_{i}}\}_{i\in J}$ and $M_{J}=\{ \emptyset \neq K\subseteq S:   \mathcal{O}_{J}=\mathcal{O}_{K}\}$ for all  $\emptyset \neq J\subseteq S$. Then the following statements hold:

\begin{enumerate}

\item If $\emptyset \neq J\subseteq S$,  $J$ is maximal in $M_{J}$ (ordered by contention) if and only if $\forall j\in S-J$, $\frac{m_{j}}{l_{j}} \nmid \mathcal{O}_{J}$.\\  
 
\item If $X:=\{ \emptyset \neq J\subseteq S: \mathcal{O}_{J}< |\codeC|  \wedge  J \text{ is maximal in } M_{J} \}\cup \{ \emptyset \}$, then $Y:=\{ J'\subseteq [n]: \exists c\in \codeC -\{1\} \text{ such that } J'=[n]-Supp(c) \}$ is equal to the set $\{J\cup([n]-S): J\in X \}$. Moreover, $d_{H}(\codeC)=n- \max \{ |J|: J\in Y \}$ and $\Delta_{H}(\codeC)=n- \min \{ |J|: J\in Y \}$.\\ 
 
\item $|\codeC|= \mathcal{O}_{S}$ and $\codeC \subseteq \widehat{G}$ where  $\widehat{G}$ is the subgroup of $G$ generated by  $\{g_{i}^{l_{i}}\}_{i\in S}$. In addition, if $g=g_{0}^{k_{1}}$ and $h=g_{0}^{k_{2}}$ are in $\codeC$, then

\[ d(g,h)=\sum_{i\in Supp(gh^{-1})}l_{i}|k_{i1}-k_{i2}|\]
 where $k_{ij}$ is the remainder of dividing $c_{i}k_{j}$ by $\frac{m_{i}}{l_{i}}$ for $i\in S$ and $j=1,2$.\\

\item Let $\widehat{d}_{L}$ and $\widehat{d}$ denote the Lee and Manhattan distance in $\widehat{G}$ (with $\widehat{G}$ as in part $3$) with respect to the generating set $\{g_{i}^{l_{i}}\}_{i\in S}$. Then for all $g,h\in \codeC$
\[l\cdot d_{H}(g,h)\leq l\cdot \widehat{d}_{L}(g,h)\leq \max\{ d_{L}(g,h), l\cdot \widehat{d}(g,h)\}\leq d(g,h)\leq \sum_{i\in S} \left( m_{i}-l_{i}\right).\]
In particular, 
\[l\cdot d_{H}(\codeC)\leq l\cdot \widehat{d}_{L}(\codeC)\leq \max\{ d_{L}(\codeC), l\cdot \widehat{d}(\codeC) \}\leq d(\codeC)\]
 
and $\Delta(\codeC) \leq  \sum_{i\in S} \left( m_{i}-l_{i}\right)$.

\end{enumerate}

\end{theorem}

\begin{proof}

\begin{enumerate}

\item Let $\emptyset \neq J\subseteq S$. Suppose that  $J$ is not maximal in $M_{J}$ (ordered by contention), i.e., $\exists J'\in M_{J}$ such that $J\subsetneq J'$. Thus, $\exists u\in J'-J\subset S-J$ such that $\frac{m_{u}}{l_{u}} \mid \mathcal{O}_{J'}=\mathcal{O}_{J}$. Conversely, let $u\in S-J$ such that $\frac{m_{u}}{l_{u}} \mid \mathcal{O}_{J}$. Then, if $J'=J\cup \{ u\}$, $\mathcal{O}_{J'}=\mathcal{O}_{J}$, implying that $J'\in M_{J}$ and $J\subsetneq J'$. Therefore $J$ is not maximal in $M_{J}$.\\

\item  Let $X:=\{ \emptyset \neq J\subseteq S: \mathcal{O}_{J}< |\codeC|  \wedge  J \text{ is maximal in } M_{J} \}\cup \{ \emptyset \}$, and $Y:=\{ J'\subseteq [n]: \exists c\in \codeC -\{1\} \text{ such that } J'=[n]-Supp(c) \}$. Let $Z:=\{J\cup([n]-S): J\in X \}$ and $J'\in Z$, then $J'=J\cup([n]-S)$ for some $J\in X$. If $J=\emptyset$, then $J'\in Y$. If $J\neq \emptyset$, since $\mathcal{O}_{J}<|\codeC|$, $g_{0}^{\mathcal{O}_{J}}\neq 1$. Besides, by construction of $\mathcal{O}_{J}$,  $g_{j}^{e_{j}\mathcal{O}_{J}}=1$ for all $j\in J$. In addition, $J$ is maximal in $M_{J}$ so that $\forall j\in S-J$, $\frac{m_{j}}{l_{j}}\nmid \mathcal{O}_{J}$ (by part $1$), implying that  $g_{j}^{e_{j}\mathcal{O}_{J}}\neq 1$ for all $j\in S-J$. Therefore $J=([n]-Supp(g_{0}^{\mathcal{O}_{J}}))\cap S$, so that
\begin{eqnarray*}
[n]-Supp(g_{0}^{\mathcal{O}_{J}})&=& \left(([n]-Supp(g_{0}^{\mathcal{O}_{J}}))\cap S\right)\cup \left( ([n]-Supp(g_{0}^{\mathcal{O}_{J}}))\cap ([n]-S)\right)\\
                                 &=&J\cup ([n]-S)=J',
\end{eqnarray*}

\noindent where the penultimate equality is because $[n]-S\subseteq [n]-Supp(g_{0}^{\mathcal{O}_{J}})$. Hence, since $g_{0}^{\mathcal{O}_{J}}\in \codeC-\{1\}$,  $J'\in Y$, and so $Z\subseteq Y$. Conversely, let $J'\in Y$ and $c\in \codeC-\{1\}$ be such that $J'=[n]-Supp(c)$, then $J'= (J'\cap S)\cup ( J'\cap ([n]-S))=(J'\cap S)\cup ([n]-S) $. Let $J=J'\cap S$. If $J=\emptyset$, then $J'\in Z$. On the other hand, since $J=([n]-Supp(c))\cap S$, if $J\neq \emptyset$ and $c= g_{0}^{k}=(g_{i}^{e_{i}k})_{i=1}^{n}$ for some $0\leq k\leq    |\codeC|-1$, then $g_{j}^{e_{j}k}=1$ for all $j\in J$ and $g_{j}^{e_{j}k}\neq 1$ for all $j\in S-J$, i.e., $\frac{m_{j}}{l_{j}}\mid k$ for all $j\in J$ and $\frac{m_{j}}{l_{j}}\nmid k$ for all $j\in S-J$. This implies that $\mathcal{O}_{J}\mid k$ and so $\frac{m_{j}}{l_{j}}\nmid \mathcal{O}_{J}$ for all $j\in S-J$. Thus $J$ is maximal in $M_{J}$ (by part $1$) and  $\mathcal{O}_{J}<\mathcal{O}_{S}=|\codeC|$. Hence $J\in X$, and so $J'\in Z$. Therefore $Z=Y$.\\

Note that $\{d_{H}(g,h):  g,h\in \codeC \wedge \, g\neq h  \}=\{|Supp(g)|: g\in \codeC-\{1\} \}$. Thus $d_{H}(\codeC)=\min \{|Supp(g)|: g\in \codeC-\{1\} \}$. Let $c_{0} \in \codeC$ be such that $d_{H}(\codeC)=|Supp(c_{0})|$, then $[n]-Supp(c_{0})\in Y$ and has maximum size because $Supp(c_{0})$ has minimum size. Thus,  $d_{H}(\codeC)=|Supp(c_{0})|= n-|[n]-Supp(c_{0})|= n-\max \{|J|: J \in Y \}$. Similarly, $\Delta_{H}(\codeC)=n-\min \{|J|: J \in Y \}$.\\

\item Since $G$ is the internal direct product of the groups $\langle g_{i} \rangle's$, $o(g_{0})=lcm\{o(g_{i}^{e_{i}})\}_{i\in S}=lcm\{\frac{m_{i}}{l_{i}}\}_{i\in S}=\mathcal{O}_{S}$. Let $\widehat{G}$ be the subgroup of $G$ generated by  $\{g_{i}^{l_{i}}\}_{i\in S}$, and $g= g_{0}^{k_{1}}$ and $h=g_{0}^{k_{2}}$ be in $\codeC$. Observe that 
\[g_{0}^{k_{j}}=\left(\prod_{i=1}^{n}g_{i}^{e_{i}}\right)^{k_{j}}=\prod_{i\in S}(g_{i}^{l_{i}c_{i}})^{k_{j}}
            =\prod_{i\in S}(g_{i}^{l_{i}})^{c_{i}k_{j}}=\prod_{i\in S}(g_{i}^{l_{i}})^{k_{ij}}=\prod_{i\in S}g_{i}^{l_{i}k_{ij}}\in \widehat{G}
\]
where $k_{ij}$ is the remainder of dividing $c_{i}k_{j}$ by $o(g_{i}^{l_{i}})=\frac{m_{i}}{l_{i}}$  $\forall i\in S$ and $j =1,2$. Thus $\codeC\subseteq \widehat{G}$. In addition, since $0\leq k_{ij} \leq \frac{m_{i}}{l_{i}}  -1$, then $l_{i}k_{ij} \leq  l_{i}(\frac{m_{i}}{l_{i}}-1)=m_{i}-l_{i}\leq m_{i}-1$ for all $i$ and $j$. Therefore  
\begin{align*}
d(g,h)&=d\left(\prod_{i\in S}g_{i}^{l_{i}k_{i1}}, \prod_{i\in S}g_{i}^{l_{i}k_{i2}}\right)=\sum_{i\in S} |l_{i}k_{i1}-l_{i}k_{i2}|\\
      &=\sum_{i\in S}l_{i}|k_{i1}-k_{i2}|=\sum_{i\in Supp(gh^{-1})} l_{i}|k_{i1}-k_{i2}|,
\end{align*}
\noindent where the sum of the last equality has all its summands distinct from zero because $|k_{i1}-k_{i2}|\neq 0$ if and only if $i\in Supp(gh^{-1})\subseteq S$.\\

\item  Let $g,h\in \codeC$ with  $g=g_{0}^{k_{1}}$ and $h=g_{0}^{k_{2}}$, and  $k_{ij}$ be the remainder of dividing $c_{i}k_{j}$ by $\frac{m_{i}}{l_{i}}$  $\forall i\in S$ and $j=1,2$. Then  $g=\prod_{i\in S}g_{i}^{l_{i}k_{i1}}$, $h=\prod_{i\in S}g_{i}^{l_{i}k_{i2}}$. Let  $z=\{g_{i}^{l_{i}}\}_{i\in S}$ and  $Supp_{z}(u)$ denote the support of $u$ with respect to $z$ for all $u\in \widehat{G}$. Now it will be shown that if $u\in \widehat{G}$, $Supp(u)=Supp_{z}(u)$. Let $u\in \widehat{G}$ be such that $u=\prod_{i\in S} (g_{i}^{l_{i}})^{u_{i}}=\prod_{i\in S} g_{i}^{l_{i}u_{i}}$ with $u_{i}\in[0,o(g_{i}^{l_{i}})-1]$  for all $ i\in [n]$, then  $l_{i}u_{i}\in [0, m_{i}-l_{i}]\subseteq[0, m_{i}-1]$ for all $ i\in [n]$. On the other hand, $u_{i}\neq 0$ if and only if $l_{i}u_{i}\neq 0$ for all $ i\in [n]$ (because $l_{i}\neq 0$ for all $i\in[n]$). Implying that $Supp(u)=Supp_{z}(u)$ and so $\widehat{d}_{H}=d_{H}$. Then 
\begin{align*}
l \cdot d_{H}(g,h)&\leq l\cdot \widehat{d}_{L}(g,h)= \sum_{i\in Supp_{z}(gh^{-1})} l\cdot \min\{|k_{i1}-k_{i2}|, o(g_{i}^{l_{i}}) -|k_{i1}-k_{i2}|\}\\
                  &\leq \sum_{i\in Supp_{z}(gh^{-1})} l_{i} \cdot \min\{|k_{i1}-k_{i2}|, \frac{m_{i}}{l_{i}}-|k_{i1}-k_{i2}|\}\\
                  &=\sum_{i\in Supp(gh^{-1})}  \min\{l_{i}|k_{i1}-k_{i2}|, m_{i}-l_{i}|k_{i1}-k_{i2}|\}=d_{L}(g,h)\leq d(g,h)
\end{align*}
\noindent where the last equality is by the form of $g$ and $h$ and the definition of the Lee distance in $G$ (with respect to the generating set $\{g_{i}\}_{i=1}^{n}$). In addition, since $l\cdot \widehat{d}_{L}(g,h)\leq l\cdot \widehat{d}(g,h)=\sum_{i\in Supp(gh^{-1})}l|k_{i1}-k_{i2}|\leq d(g,h)$ (where the last inequality is by part $3$), then 
\[l\cdot d_{H}(g,h)\leq l\cdot \widehat{d}_{L}(g,h)\leq \max\{ d_{L}(g,h), l\cdot \widehat{d}(g,h)\}\leq d(g,h).\]
On the other hand,  
\[
d(g,h)=\sum_{i\in Supp(gh^{-1})} l_{i}|k_{i1}-k_{i2}|
      \leq  \sum_{i\in Supp(gh^{-1})} l_{i} \left(\dfrac{m_{i}}{l_{i}}-1\right)\leq \sum_{i\in S} \left( m_{i}- l_{i} \right), 
\]
\noindent where the first inequality is because $k_{i1}, k_{i2}\in [0, \frac{m_{i}}{l_{i}}-1 ]$, and the last one is because $Supp(c)\subseteq S$  $\forall c\in \codeC $. Thus $ \Delta(\codeC) \leq  \sum_{i\in S} \left( m_{i}- l_{i} \right)$. Therefore  
\[l\cdot d_{H}(\codeC)\leq l\cdot \widehat{d}_{L}(\codeC)\leq \max\{ d_{L}(\codeC), l\cdot \widehat{d}(\codeC) \}\leq d(\codeC)\]
and $\Delta(\codeC) \leq  \sum_{i\in S} \left( m_{i}-l_{i}\right).$  

\end{enumerate}

\end{proof}

\begin{example}

Let $G\simeq C_{8}\times C_{8}\times C_{8} \times C_{8}$, then there exists a generating set $\{ g_{i} \}_{i=1}^{4}$ (where $o(g_{i})=8$ for $i=1,...,4$) such that any element of $G$ can be uniquely written in terms of the $g_{i}'s$. Let $g_{0}=g_{1}^{2}g_{2}^{2}g_{3}^{4}g_{4}^{4}$ and $\codeC=
\langle g_{0}\rangle=\{ 1, g_{1}^{2}g_{2}^{2}g_{3}^{4}g_{4}^{4}, g_{1}^{4}g_{2}^{4}, g_{1}^{6}g_{2}^{6}g_{3}^{4}g_{4}^{4}\}$. Let $S=Supp(g_{0})=\{1,2,3,4\}$. Then $X=\{\emptyset \neq J\subseteq S: \mathcal{O}_{J}<4 \wedge J \text{ is maximal in }$ $ M_{J}\}\cup\{ \emptyset \}=\{ \{3,4\}, \emptyset \}$, so that $d_{H}(\codeC)=4-\max\{|J|: J\in Y\}=4-|\{3,4\}|=2$ (by Theorem \ref{cyclic-code}, part $2$). Since $l_{1}=l_{2}=gcd(8,2)=2$ and  $l_{3}=l_{4}=gcd(8,4)=4$, then $l=\min\{l_{i}\}_{i\in S}=\min\{2,4\}=2$. In addition, since the code $\codeC$ is isometric to the code $\upsilon (\codeC)=\{0000,2244,4400,6644\}$ (where $\upsilon$ is defined as at the beginning of this section) in $\mathfrak{G}=[0,7]^{4}$, then $d(\codeC)=d(\upsilon (\codeC))=8$. Thus,
\begin{align*}
l\cdot d_{H}(\codeC)&= 2\cdot 2=4\leq  l\cdot \widehat{d}_{L}(\codeC)=2\cdot 4=8\\
                    &\leq \max\{ d_{L}(\codeC), l\cdot \widehat{d}(\codeC) \}=\max\{ 8, 8\}=8 \leq d(\codeC)=8,
\end{align*}
and $\Delta(\codeC)=20 \leq  \sum_{i=1}^{4} \left( 8-l_{i}\right)=(8-2)+(8-2)+(8-4)+(8-4)=20$, which coincides with Theorem \ref{cyclic-code} (part $4$).  \end{example}

\section*{Conclusion}\label{conclusion}\label{S7}

In this work, we introduce the concept of grid code as a subset of $\mathfrak{G}:=\prod_{i=1}^{n}[0,m_{i} -1]$ with the Manhattan distance, and present some lower and upper bounds on the maximum size that a grid code with a prescribed minimum distance could have.  In addition, a relation between grid codes and algebraic coding theory is presented. By exploring some relations between the Hamming, Lee, and Manhattan distances, some lower bounds for the minimum Manhattan distance in terms of the minimum Hamming and Lee distances of grid codes that are cyclic subgroups of an abelian group are provided.\\

Since grid codes present an alternative to codes with the Hamming distance, from a theoretical perspective, it is appealing to explore what classical results or questions can be translated to this new context, as we did with the Hamming and Gilbert-Varshamov bounds. For instance, Asking if it is possible to get asymptotic versions of the Hamming and Gilbert-Varshamov bounds for grid codes would be natural. Recently, in \cite{Byrnea} an asymptotic version of the Gilbert-Varshamov bound for random codes (those generated by random matrices) over a finite chain ring with general metrics (that can be extended on tuples additively) was presented,  and it was proven that random codes achieve the asymptotic Gilbert-Varshamov bound with high probability. Thus, if $p$ is a prime number, one may ask if random codes in $(\mathbb{Z}_{p^{k}})^{n}$ (with the Manhattan distance) meet the Gilbert-Varshamov bound with high probability.\\

The theory presented in this paper could serve as a reference to model and tackle problems of research areas in which $n$-dimensional grids and grid codes appear. The following hypothetical application is part of a theoretical exercise to illustrate the bounds' potential proved in this work. Imagine that a material with a $3$-dimensional  grid shape  (as the ones mentioned in the Introduction \ref{S1}) were studied and the designers needed some special nodes on that grid satisfying certain properties (e.g., physical or chemical)  that were not compatible with these nodes being too close together (maybe the material would be too expensive to produce or could lose another desired property if the nodes were too close, such as electrical or thermal conductivity, flexibility, or resistance) so these nodes should satisfy that the minimum distance between them is certain positive integer number $d$. In that context, the ambient space of fixed dimensions would be the material with a $3$-dimensional grid shape and the set of special nodes would be a grid code with minimum distance $d$. Thus if the designers wanted to build this material containing a code with minimum distance $d$ and maximum possible size,  the Hamming and Gilbert-Varshamov bounds may be used to know an interval where the maximum size should belong.



\section*{Acknowledgements}

The second author would like to thank the Institute for Algebra and Geometry at Otto von Guericke University - Magdeburg, where part of this work was carried out, for the hospitality displayed. Also acknowledges and thanks the Deutscher Akademischer Austauschdienst for the financial support provided.


\begin{thebibliography}{99}











\bibitem{Ahlswede} R. Ahlswede, ``Group Codes do not Achieve Shannon's Channel Capacity for General Discrete Channels,'' The Annals of Mathematical Statistics, vol. 42, no. 1, pp. 224-240, 1971.

\bibitem{Araujo} C. Araujo, I. Dejter, P. Horak, ``A generalization of Lee codes,'' Des. Codes Cryptogr., vol. 70, pp. 77–90, 2014.

\bibitem{Atola82} J. Astola, ``An Elias-Type Bound for Lee Codes over Large
Alphabets and its Application to Perfect Codes,'' IEEE Trans. on Inf. Theory, vol. 28, pp. 111-113, 1982. 

\bibitem{Atola84} J. Astola, ``On the asymptotic behaviour of Lee-codes,'' Discrete Applied Mathematics, vol. 8, pp. 13-23, 1984. 

\bibitem{Berlekamp} E. R. Berlekamp,\textit{ Algebraic Coding Theory}, Revised ed., CA Aegean Park Press, Laguna Hills, 1984.

\bibitem{Bevan}D. Bevan, C. Homberger, B. E. Tenner, ``Prolific permutations and permuted packings: Downsets containing many large patterns,'' Journal of Combinatorial Theory, Series A, vol. 153, pp. 98–121, 2018.

\bibitem{Biglieri} E. Biglieri and M. Elia, ``Construction of of Linear Block Codes Over Groups," Proceedings. IEEE International Symposium on Information Theory, 1993, pp. 360-360, doi: 10.1109/ISIT.1993.748676.

\bibitem{Byrnea} E. Byrnea, A. Horlemannb, K. Khathuriac, V. Wegera, ``Density of free modules over finite chain rings'' Linear Algebra and its Applications, vol. 651, pp. 1–25, 2022.

\bibitem{Blackburn} S. R. Blackburn, C. Homberger, P. Winkler, ``The minimum Manhattan distance and minimum jump of permutations,'' Journal of Combinatorial Theory, Series A, vol. 161, pp. 364–386, 2019.

\bibitem{Material1} S. Cordero-Sánchez, F. Rojas-González, G. Román-Alonso, M.A. Castro-García, M. Aguilar-Cornejo, J. Matadamas-Hernández, ``Pore networks subjected to variable connectivity and geometricalrestrictions: A simulation employing a multicore system,'' Journal of Computational Science, vol. 16, pp. 177–189, 2016.

\bibitem{Deza1} M. M. Deza and E. Deza, \textit{Encyclopedia of Distances}, Springer-Verlag Berlin Heidelberg, 2009.

\bibitem{Deza2} M. M. Deza and E. Deza, \textit{Dictionary of Distances} (3rd ed.), Elsevier, 2014.

\bibitem{Elias} P. Elias. ``Coding for two noisy channels,'' IRE Convention Record, Part 4, pp. 37-47, 1955.

\bibitem{Etzion} T. Etzion, ``Product constructions for perfect Lee Codes,'' IEEE Trans. Inf. Theory, vol. 57, no. 11, pp. 7473-7481, 2011.


\bibitem{Etzion2} T. Etzion, A. Vardy, E. Yaakobi, ``Coding for the Lee and Manhattan Metrics With Weighing Matrices,'' IEEE Trans. Inf. Theory, vol. 59, no. 10, pp. 6712-6723, 2013.




\bibitem{Forney92} G. D. Forney, ``On the Hamming Distance Properties of Group Codes,'' IEEE Trans. on Inf. Theory, vol. 38, no. 6, pp. 1797-1801, 1992.

\bibitem{Forney93} G. D. Forney, M. D. Trott, ``The Dynamics of Group Codes: State Spaces, Trellis Diagrams, and Canonical Encoders,'' IEEE Trans. on Inf. Theory,  vol. 39, no. 9, pp. 1491-1513, 1993.

\bibitem{Gabidulin} E. Gabidulin, ``A brief survey of metrics in coding theory,'' Mathematics of Distances and Applications, vol. 66,
 pp. 66-84, 2012.
 
\bibitem{garcia-tapia} E. J. Garc\'ia-Claro, H. Tapia-Recillas, On the dimension of ideals in group algebras, and group codes, Journal of Algebra and Its Applications, 2022, vol. 21, no 02, p. 2250024. Online Ready, https://doi.org/10.1142/S0219498822500244.



\bibitem{Golomb} S.W. Golomb, L.R. Welch. ``Perfect codes in the Lee metric and the packing of polyominos,'' SIAM J. Appl. Math, vol. 18, pp. 302-317, 1970.

\bibitem{Greche} L. Greche, M. Jazouli, N. Es-Sbai, A. Majda and A. Zarghili, ``Comparison between Euclidean and Manhattan distance measure for facial expressions classification," 2017 International Conference on Wireless Technologies, Embedded and Intelligent Systems (WITS), 2017, pp. 1-4, doi: 10.1109/WITS.2017.7934618.

\bibitem{Huffman} W. C. Huffman, V. Pless, \textit{Fundamentals of Error-Correcting Codes}, Cambridge University Press, 2003.


\bibitem{Interlando} J. C. Interlando, R. Palazzo and M. Elia, ``Group block codes over nonabelian groups are asymptotically bad," in IEEE Trans. on Inf. Theory, vol. 42, no. 4, pp. 1277-1280, July 1996, doi: 10.1109/18.508859.

\bibitem{Krithivasan} D. Krithivasan and S. S. Pradhan, ``Distributed source coding using Abelian group codes: Extracting performance from structure," 2008 46th Annual Allerton Conference on Communication, Control, and Computing, 2008, pp. 1538-1545, doi: 10.1109/ALLERTON.2008.4797745.

\bibitem{Lee} C. Y. Lee ``Some properties of nonbinary error-correcting code right,'' IRE Trans. Inf. Theory, vol. 4, pp. 72-82, 1958.

\bibitem{Li} Z. Li, X. Zhong, J. Wei and H. Shi, ``The Application of Manhattan Tangent Distance in Outdoor Fingerprint Localization," 2018 IEEE Global Communications Conference (GLOBECOM), 2018, pp. 1-5, doi: 10.1109/GLOCOM.2018.8647175.

\bibitem{Massey} P. Massey, ``Many Non-Abelian Groups Support Only Group Codes That Are Conformant To Abelian Group Codes,'' IEEE Trans. on Inf. Theory. ISIT, 1997. Ulm. Germany.

\bibitem{Material2} J. Matadamas-Hernández, G. Román-Alonso, F. Rojas-González, M.A. Castro-García, Azzedine Boukerche, M. Aguilar-Cornejo, and S. Cordero-Sánchez, ``Parallel Simulation of Pore Networks
Using Multicore CPUs,'' IEEE trans. on comp., vol. 63, no.6, pp. 1513-1525, 2014.

\bibitem{Cristal1} A. McPherson, \textit{Introduction to Macromolecular Crystallography}, Wiley-Blackwell by John Wiley \& Sons, Inc, 2009.

\bibitem{Roth} R. Roth, \textit{Introduction to Coding Theory}, Cambridge University Press, 2006.

\bibitem{Richardson} T. Richardson, \textit{Modern Coding Theory}, Cambridge University Press, 2008.

\bibitem{Sahebi1} A. G. Sahebi and S. S. Pradhan, ``On the capacity of Abelian group codes over discrete memoryless channels," 2011 IEEE International Symposium on Information Theory Proceedings, 2011, pp. 1743-1747, doi: 10.1109/ISIT.2011.6033846.

\bibitem{Sahebi2} A. G. Sahebi and S. S. Pradhan,``Codes over non-Abelian groups: Point-to-point communications and computation over MAC," 2012 IEEE International Symposium on Information Theory Proceedings, 2012, pp. 631-635, doi: 10.1109/ISIT.2012.6284269.

\bibitem{Cristal2}D. E. Sands, \textit{Introduction to Crystallography}, Dover Publications, Inc, 1993.


\bibitem{Sok} L. Sok, J. Belfiore, P. Solé, A. Tchamkerten, ``Lattice Codes for Deletion and Repetition Channels,'' IEEE Trans. on Inf. Theory, vol. 64, pp. 1481-1496, 2018.


\bibitem{tucker} A. Tucker, \emph{Applied Combinatorics}, 6th.\ edition, John Wiley and Sons, 2012.


\bibitem{Ulrich} W. Ulrich, ``Non-binary error-correcting codes,'' Bell Syst. Tech. J.,  vol. 36, pp. 1341-1387, 1957.

\bibitem{Zeulin}N. V. Zeulin, I. A. Pastushok, A. M. Turlikov and V. A. Davydov, ``On Performance of Remainder Codes in Manhattan Metrics over AWGN Channel with Modulation 1024-QAM," 2019 Wave Electronics and its Application in Information and Telecommunication Systems (WECONF), 2019, pp. 1-4, doi: 10.1109/WECONF.2019.8840605.









 




 










\end{thebibliography}
\end{document}